\documentclass[11pt]{amsart}
\usepackage[T1]{fontenc}
\usepackage[english]{babel}
\usepackage{lmodern,microtype}
\usepackage{amsmath,amsfonts,amsthm,mathtools}
\usepackage{graphicx}
\usepackage{xcolor,fullpage}
\usepackage{tikz}
 \usepackage{faktor}
\usetikzlibrary{calc,arrows.meta}
\usepackage{placeins}
\usepackage{csquotes}
\usepackage[
    style=numeric,
    sorting=nyt,
    backend=biber,
    maxbibnames=99
]{biblatex} 
\usepackage{hyperref}
\hypersetup{
    colorlinks,
    linkcolor={red!50!black},
    citecolor={blue!50!black},
    urlcolor={blue!80!black}
}
\usepackage[capitalise,nameinlink]{cleveref}
\usepackage{subcaption}

\usepackage[left=1.3in,right=1.3in,top=1.1in,bottom=1.1in]{geometry}

\newtheorem{theorem}{Theorem}
\newtheorem{lemma}[theorem]{Lemma}
\newtheorem{claim}[theorem]{Claim}

\newtheorem*{lemma*}{Lemma}

\theoremstyle{definition}

\theoremstyle{remark}
\newtheorem*{remark}{Remark}
\theoremstyle{plain}
\usepackage{amssymb}

\usepackage{todonotes}
\definecolor{cyan}{RGB}{0,190,210}
\definecolor{magenta}{RGB}{230,0,250}
\definecolor{darkkred}{RGB}{230,0,0}
\definecolor{darkgreen}{RGB}{0,160,0}

\newcommand{\ind}[1]{\mathbf{1}_{#1}} 
\newcommand{\pyr}[2]{\operatorname{pyr}_{#1}(#2)} 

\newcommand{\cL}{\mathcal{L}}

\newcommand{\maxfn}[1]{\MAX_{#1}}

\newcommand{\upto}[1]{[#1]}

\DeclareMathOperator{\conv}{conv}
\DeclareMathOperator{\relu}{ReLU}
\DeclareMathOperator{\ICNN}{\mathsf{ICNN}}
\DeclareMathOperator{\CPWL}{\mathsf{CPWL}}

\DeclareMathOperator{\MAX}{MAX}

\DeclareMathOperator{\aff}{aff}

\DeclarePairedDelimiterX{\inner}[2]{\langle}{\rangle}{#1,#2}
\DeclarePairedDelimiter{\set}{\{}{\}}
\DeclarePairedDelimiter{\abs}{\lvert}{\rvert}
\DeclarePairedDelimiterX{\rset}[2]{\{}{\}}{#1\mathclose{}:\mathopen{}#2}
\DeclarePairedDelimiterX{\segment}[2]{[}{]}{#1,#2}

\newcommand{\R}{\mathbb{R}}
\newcommand{\Q}{\mathbb{Q}}
\newcommand{\Z}{\mathbb{Z}}

\colorlet{plus}{blue!70!black}
\colorlet{minus}{red!70!black}
\colorlet{accent}{orange!85!black}
\colorlet{coneA}{teal!70!black}
\colorlet{coneB}{violet!70!black}

\tikzset{
  pyramid base view/.style={
    x={(.9cm,0cm)},y={(.18cm,.85cm)},z={(0cm,.9cm)},
    line join=round},
  pyramid view/.style={
    pyramid base view,y={(.36cm,.4cm)}}}


\title{Two $\boldsymbol{\relu}$ Layers is All You Need} 

\author{Florestan Brunck}
\address{Department of Computer Science, University of Copenhagen, Denmark}
\email{flbr@di.ku.dk}

\author{Christoph Hertrich}
\address{University of Technology Nuremberg,  Germany}
\email{christoph.hertrich@utn.de}

\author{Kilian Rueß}
\address{University of Technology Nuremberg, Germany}
\email{kilian.ruess@utn.de}

\author{\\Jack Stade}
\address{Department of Computer Science, University of Copenhagen, Denmark}
\email{jast@di.ku.dk}

\author{Amir Yehudayoff}
\address{Department of Computer Science, University of Copenhagen, Denmark
and Department of Mathematics, Technion--IIT}
\email{amye@di.ku.dk}

\begin{document}
\begin{abstract}
We prove that every $\relu$ network can be exactly simulated by a $\relu$ network with two hidden layers. We present two proofs: one that is based on fillings in homological algebra and a direct combinatorial proof. 
\end{abstract}
\maketitle

\section{Introduction}
The ubiquity and success of neural networks as a   model of computation motivates studying the capabilities of their architectures.
This study raises familiar complexity-theoretic questions, like what is the class of functions they can compute, at what depth, with how many gates, etc.? 

This work investigates the role of depth in neural networks.
On the one hand, deep learning is one of the major successes of modern computer science.
On the other hand, classical approximation theorems dating back to the late 1980s show that, for many activation functions including $\relu$, a single hidden layer suffices to approximate any continuous function arbitrarily well on a compact domain~\cite{cybenko1989approximation,hornik1989multilayer,leshnoMultilayerFeedforwardNetworks1993}.
Although various explanations for the power of compositions have been proposed~\cite{bengio2009learning,mhaskar2017when}, a quantitative understanding of the role of depth remains incomplete~\cite{huchette2026when,hertrichLowerBoundsDepth2023}.

The role of depth has been studied from several perspectives.
One line of work investigates depth-width tradeoffs in expressivity: can additional layers reduce the resources needed to represent or approximate a function~\cite{aroraUnderstandingDeepNeural2018,eldan2016power,telgarsky2016benefits}?
Another asks whether depth can provide advantages for training~\cite{safran2022optimization}.
Perhaps the most fundamental question, however, is whether deeper networks can represent functions that shallower networks cannot.
Until now, this question has remained open.

We focus on neural networks that use the $\relu(x) = \max\{0,x\}$ activation function, which are among the most widely used ones in practice~\cite{glorot2011deep,huchette2026when}. 
A $\relu$ network with $k$ hidden layers (and depth $k+1$) is traditionally defined as the alternating composition
\[
T^{(k)}\circ \relu \circ T^{(k-1)} \circ \cdots \circ  \relu \circ T^{(1)} \circ \relu \circ{} T^{(0)}
\]
where $T^{(i)}$ is an affine map between two Euclidean spaces and $\relu$ stands for entry-wise application.

$\relu$ networks compute continuous and piecewise linear (CPWL) functions.
These are continuous functions that are affine on each cell of a finite polyhedral partition of the input Euclidean space.
Conversely, every CPWL function can be computed by a $\relu$ network (see e.g.~\cite{arora2018understanding}).
Their ability to capture nonlinear behavior while retaining a simple local affine structure makes CPWL functions useful beyond neural networks.
For example, they are central to finite element methods~\cite{he2020relu}, and their structure enables exact verification of stability conditions in robotic control~\cite{dai2021lyapunov}.


For a ring $R \subseteq \R$, denote by $\CPWL_n(R)$ the space of $\CPWL$ functions whose affine pieces have coefficients in $R$. Denote by $\relu_{n,k}(R)$ the subset of $\CPWL_n(R)$ representable by $\relu$ networks with $k$ hidden layers, and weights and biases in $R$. 
The obvious inclusions are
\begin{align}
\relu_{n,0}(R)
\subseteq \relu_{n,1}(R)
\subseteq \relu_{n,2}(R)
\subseteq \ldots \subseteq \CPWL_n(R). \label{eqn:chain}    
\end{align}
Our main results work when $R$ is any field containing $\Q$. We often write $\CPWL_n$ and $\relu_{n,k}$ for $\CPWL_n(\Q)$ and $\relu_{n,k}(\Q)$ respectively. 

The main goal of this work is understanding the inclusions in the chain~\eqref{eqn:chain}.
Wang and Sun~\cite{wang2005generalization} identified a key structural result; see also \cite{hertrichLowerBoundsDepth2023,bakaev2026better} for alternative, more geometric proofs. They showed that 
\[\MAX_n(x_1,\ldots,x_n) \coloneqq \max \set{x_1,\ldots,x_n}\]
is complete for $\CPWL_{n-1}(\R)$ with respect to depth.
That is, every function $f \in \CPWL_n(\R)$ can be expressed as a finite sum
\begin{align}
\label{eqn:maxn}
f(x) = \sum_i \sigma_i \MAX_{n+1}(A_i(x))
\end{align}
where $A_i:\R^n \to \R^{n+1}$ is an affine map and $\sigma_i \in \{\pm 1\}$.
The geometric proof by~\cite{bakaev2026better} for the universality of $\MAX_{n+1}$ relies on the fact that $\MAX_{n+1}$ is the support function of a simplex, and every polytope can be triangulated into simplices.

The first inclusion in~\eqref{eqn:chain} is trivially strict because 
$\relu_{n,0}(\R)$ is precisely the set of affine functions.
Mukherjee and Basu~\cite{mukherjee2017lower} 
showed that the second inclusion is strict for $n \geq 3$ because $\MAX_3 \in \relu_{3,2}(\R)$ but $\MAX_3 \not \in \relu_{3,1}(\R)$.
In the other direction, 
Arora, Basu, Mianjy, and Mukherjee~\cite{arora2018understanding} observed that \eqref{eqn:maxn} implies that the length of the chain is at most logarithmic: 
\[\CPWL_n = \relu_{n,k} \qquad \text{for $k \geq \lceil \log_2 (n+1) \rceil$.}\]
The reason is that a divide-and-conquer network
of depth $\lceil \log_2(n+1)\rceil$ computes $\MAX_{n+1}$.
Subsequently, Hertrich, Basu, Di Summa, and Skutella~\cite{hertrichLowerBoundsDepth2023} conjectured that the length of the chain is exactly $\lceil \log_2 (n+1) \rceil$. In other words, they conjectured that $\MAX_{n}$ can not be computed by less than $\lceil \log_2 (n) \rceil$ hidden layers. 

A series of works validated this conjecture in a variety of special cases. 
Haase, Hertrich, and Loho~\cite{haase2023lower} proved the conjecture over the integers:
\[k < \lceil \log_2(n) \rceil
\qquad \Rightarrow \qquad \MAX_n \not \in  \relu_{n,k}(\Z).\]
Averkov, Hojny, and Merkert~\cite{averkov2025expressiveness} extended the logarithmic lower bounds to rings of the form $\Z[1/N]$ for an integer $N>0$. They showed that if a prime $p$ does not divide $N$ then
\[k < \lceil \log_p(n) \rceil
\qquad \Rightarrow \qquad \MAX_n \not \in \relu_{n,k}(\Z[1/N]).\]
There are more lower bounds for special classes of networks. 
Valerdi~\cite{valerdi2024minimal} proved the conjecture for input convex neural networks ($\ICNN$), which was defined by Amos, Xu and Kolter~\cite{amos2017input} and studied in subsequent works (e.g.~\cite{bunning2021input,chen2018optimal,chen2020data,bunne2023learning,hertrich2024neural}). 
Grillo, Hertrich, and Loho~\cite{grillo2025depth} proved an $\Omega( \log \log n)$ lower bound for $\relu$ networks that ``respect the singularities of $\MAX_n$''.

The conjecture was, however, recently disproved by Bakaev, Brunck, Hertrich, Stade, and Yehudayoff~\cite{bakaev2026better}. They proved that
\[\MAX_5 \in \relu_{5,2}\]
and concluded the asymptotic improvement 
\[\MAX_n \in \relu_{n,k}
\qquad \text{for $k = \lceil\log_3(n-2)\rceil+1$}.\]
Shortly after, computer-assisted constructions~\cite{ruess2026shallower} showed that
\[\MAX_{12} \in \relu_{12,2}.\]
Again, this finite case can be converted into an improvement of the asymptotic bound from~\cite{bakaevBetterNeuralNetwork2026}, yielding
\[\MAX_n \in \relu_{n,k}
\qquad \text{for $k = \lceil\log_6(n/2)\rceil+1$}.\]

Our main result is that 
the entire depth hierarchy collapses to the second level, which is optimal by~\cite{mukherjee2017lower}. See Figure~\ref{fig:hierarchy} for an illustration.

\begin{theorem}\label{thm:main}
  For all $n\geq 1$ and any field $R\subseteq \mathbb{R}$ containing $\Q$,
  \[\CPWL_n(R) =\relu_{n,2}(R).\]
\end{theorem}

\begin{figure}\centering
\includegraphics[page=2,width=14cm]{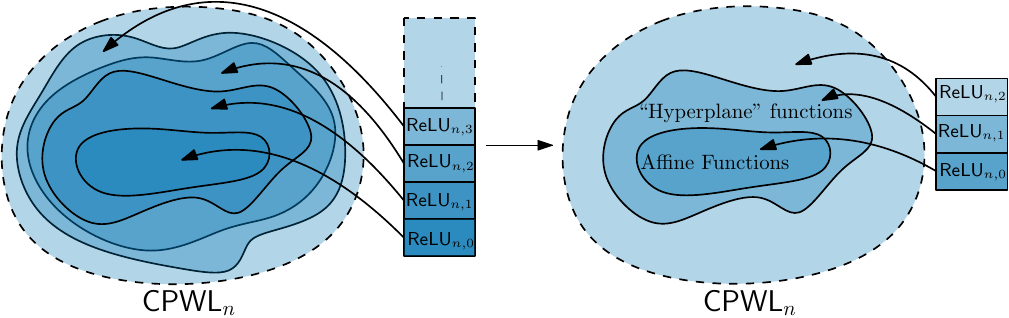}
\caption{The depth hierarchy of $\relu$ networks collapses to its first 3 levels.}
\label{fig:hierarchy}
\end{figure}
To prove our main result, it suffices to prove the inclusion \[\MAX_n\in\relu_{n,2}(\Q)\] 

We provide two proofs of the inclusion.
The first is abstract and uses the language of homological algebra. 
The second is concrete and is based on exact combinatorial cancellations. 

We, in fact, show that there is a $\relu$ network with two hidden layers for $\MAX_n$ over the rationals where the precision required is 
$\text{poly}(n)$ many bits.
The lower bound from~\cite{averkov2025expressiveness} says that
such a network must divide by the product of all primes up to $\approx \sqrt{n}$, which requires precision of $\approx \sqrt{n}$ many bits.
The combinatorial proof requires division by $n!$,
which is consistent with this barrier.
As written, the homological algebra proof also uses division by $n!$, but this is not strictly necessary. 
What the homological proof actually uses is division by $2^k-1$ for $k=1,\dots,n$, 
which is consistent with the barrier by Fermat's little theorem. 


\subsection{Building the Simplex}

Because $\MAX_{n}$ is positively homogeneous, we can assume that all biases in the network expressing it are zero, see~\cite{hertrichLowerBoundsDepth2023}. 
Because $\MAX_2(x_1,x_2) = x_1 + \relu(x_2-x_1)$,
we can and will work with networks that use $\MAX_2$ gates instead of $\relu$ gates, also known as \emph{maxout} networks.

There is a duality between computing functions with $\MAX_2$ networks and building formal Minkowski differences of polytopes using the two operations: Minkowski sums and convex hulls of a union (see e.g.~\cite{hertrichLowerBoundsDepth2023,bakaevBetterNeuralNetwork2026} and references therein). This duality says that computing $\MAX_{n+1}$ is equivalent to building the $n$-dimensional simplex $\Delta_n$. (There are $n+1$ vertices and the dimension is $n$.)

Computations with one hidden layer exactly correspond to differences of \emph{zonotopes}, i.e.\ Minkowski sums of segments.
Every zonotope can be decomposed into parallelepipeds (\cite[Theorem~(54)]{shephard1974combinatorial}) so one hidden layer is actually captured by parallelepipeds $\sum_i [a_i,b_i]$.
A general two-layer computation corresponds to the convex hull of two zonotopes.

Both of the constructions we describe use a particular class of two-layer components we call \emph{paramids}.
A paramid is the convex hull of an apex
$\gamma \in \R^n$ and a parallelepiped
$\sum_i [a_i,b_i]$.
For an example see \Cref{fig:paramid}.
The support function of a paramid is thus of the form
\[\MAX_2\Big(\langle \gamma, x\rangle, \sum_i \MAX_2(\langle a_i, x\rangle, \langle b_i, x\rangle)\Big).\]
Paramids are the simplest two-layer polytopes and are thus a natural starting point to understand two-layer computations. 
Our results imply that $\relu_{n,2}$ is spanned by support functions of paramids. Actually, there is a more direct way to see this (see \Cref{sec:paramidscomplete}).
Informally, every convex hull of two zonotopes can be expressed algebraically in terms of paramids.


\begin{figure}
    \centering
    \includegraphics[page=5]{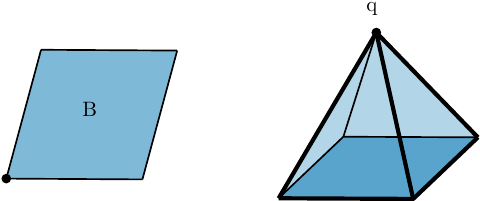}
    \caption{A paramid given as the pyramid over a two-dimensional parallelepiped $B$ with apex $q$.}
    \label{fig:paramid}
\end{figure}

Another natural idea is to work modulo lower-dimensional pieces. A $\CPWL$ function $f\colon\R^n \to \R$ is said to be lower-dimensional if there is a non-zero direction $v$ such that $f(x+tv)=f(x)$ for all $x\in \R^n$ and $t\in \R$. 
So if, for each $n$, we can write
$\MAX_{n+1}$ as a 
\[\MAX_{n+1} = \sum_i c_i h_{P_i} + R\]
where each $P_i$ is a paramid and the remainder $R$ is a sum of lower-dimensional functions, then we are done by induction.
The reason is that every lower dimensional function
is reducible to $\MAX_{n}$, which we already know how to compute.

Subdivisions give us a powerful tool in these arguments. If we can subdivide $\Delta$ into pieces that can be computed in $k$ hidden layers, then $\Delta$ can be computed by $k$ hidden layers. So subdivisions are ``free'' in the sense that they do not cost us anything in terms of depth. This follows from the following well-known inclusion-exclusion-like formula also known as valuation (see e.g. Lemma 6.1.1 in~\cite{Schneider2013} or 5.20 in~\cite{MS1983}):

\begin{lemma}[Valuation property]\label{lem:additivity}
Suppose that a polytope $X$ can be written as a union of polytopes $X=A_1\cup \dots\cup A_k$. Then,
\[X+\sum_{\substack{\emptyset \ne S\subseteq\{1, \dots, k\}\\|S| \text{ even}}}\space\bigcap_{i\in S}A_i=\sum_{\substack{S\subseteq\{1, \dots, k\}\\|S| \text{ odd}}}\space\bigcap_{i\in S}A_i,\]
where an empty polytope is interpreted as zero.
\end{lemma}

Our new goal is thus to build the simplex using paramids modulo sums of lower dimensional pieces. 
We present two such constructions.
One construction proceeds by induction and uses the language of chain complexes and homological algebra. 
The second construction is based on a direct inclusion-exclusion formula.

\subsection{Overview of the Homological Algebra Implementation}

Let us start with some intuition.
Imagine we already found a two-layer representation of $\Delta_n$.
By restricting the representation to a facet,
we get a two-layer representation of $\Delta_{n-1}$. 
Each facet has its own representation
and the $n+1$ representations are somehow glued together coherently. 
What extra conditions does the boundary representation need to satisfy in order for it to extend to the $n$-simplex? 

This type of ``boundary filling'' question naturally suggests using the framework of chain complexes and homology.
A chain complex $C_{\bullet}$ is a sequence of vector spaces $C_0, C_1, C_2, \dots$ equipped with linear boundary maps $d:C_k\rightarrow C_{k-1}$. The boundary map satisfies $d^2=0$,
which implies that the image of $d:C_{k+1}\rightarrow C_k$ is a subspace of the kernel of $d:C_k\rightarrow C_{k-1}$.
This leads to \emph{homology}:
\[H_k(C_{\bullet})=\faktor{\ker(d:C_k\rightarrow C_{k-1})}{\text{img}\left(d:C_{k+1}\rightarrow C_k\right)}.\]
An element $x$ satisfying $d(x)=0$ is called a \emph{cycle} (we use lowercase $x$ to distinguish a formal element $x$ of a chain complex from a polytope, which we would represent with a capital $X$). An element $x$ that is equal to $d(y)$ for some $y$ is called a \emph{boundary}. The condition that $d^2=0$ says that every boundary is a cycle. The homology groups record which cycles are boundaries. 

In \Cref{sec:minkowskicomplex}, we define the chain complex $P_{\bullet}$
of formal Minkowski differences of polytopes modulo lower dimensional polytopes. 
This is the main chain complex we care about,
because it corresponds to $\relu$ networks
by the duality discussed above. 
A central property for performing induction is that a top-dimensional class in $P_{\bullet}$ is determined by its boundary. 
Namely, for chains of polytopes in $\R^n$, the map $d: P_n \to P_{n-1}$ is injective.
This is a consequence of known results (see \Cref{lem:topinjectivity}). 
So if we can find some $x \in P_n$ such that $d(x) = d(\Delta_n)$ then we actually obtained $\Delta_n$.

We let $\Lambda_{\bullet}$ be the subcomplex of $P_{\bullet}$ generated by paramids. 
This chain complex corresponds to the subclass of two-layer $\relu$ networks we focus on. The homology of $\Lambda_{\bullet}$ tells us a lot about the functions that can be represented in terms of paramids.
If we can represent the $n$-simplex by paramids, then the boundary of the $(n+1)$-simplex is a cycle in $\Lambda_{n}$. By injectivity, the question of whether we can represent the $(n+1)$-simplex by paramids exactly asks whether this cycle is a boundary in $\Lambda_{\bullet}$.

The Minkowski sum relations in $\Lambda_{\bullet}$ are complicated, and it is difficult to perform calculations in this complex directly. So we try to approximate $\Lambda_{\bullet}$ with a chain complex that is easier to understand. 
We define a chain complex $C_{\bullet}$ of parallelepipeds. 
There is a coning map $\rho:C_{k}\rightarrow P_{k+1}$ that sends a parallelepiped $X$ to the paramid with base $X$ and apex at the origin. The image of $\rho$ lies in the subcomplex $\Lambda_{\bullet}$ of $P_{\bullet}$.
A parallelepiped is a sum of intervals so it is zero in $P_k$ for $k\geq 2$.
The parallelepiped facet of a $k$-dimensional paramid is zero in $P_{k-1}$ for $k \geq 3$.
Because the base does not survive in $\Lambda_\bullet$, this means that $\rho$ commutes with the boundary maps of $C_{\bullet}$ and $P_{\bullet}$ for $k \geq 2$.

Let us try to implement the inductive construction. 
We start with a chain of parallelepipeds $x\in C_{n-1}$ such that $\rho(x)$ is the $n$-simplex. 
We then create copies of $x$ for each facet of the $(n+1)$-simplex to get a cycle $x'\in C_{n-1}$ such that $\rho(x')$ is the boundary of the $n+1$ simplex. 
If we can find some $y \in C_n$ such that $d(y)=x'$, then $\rho(y)$ satisfies \[d(\rho(y))=\rho(d(y))=\rho(x')=d(\Delta_{n+1}).\]
Injectivity tells us that $\rho(y)$ is the needed $(n+1)$-simplex.

This still does not explain how we find $y\in C_{n}$ with $d(y)=x'$. We notice that not all cycles in $C_{\bullet}$ are boundaries. Every boundary $x$ in $C_{\bullet}$ satisfies the \emph{augmentation condition}, which says that $x$ becomes $0$ if we work modulo translations, because pairs of opposite faces of a parallelepiped give sign-reversed translates of each other after taking the boundary. But there are cycles in $C_{\bullet}$ that do not have this property, see for example \Cref{fig:trianglecycle}.

\begin{figure}[h]
\centering
\includegraphics[page=1]{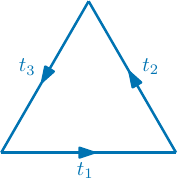}
\caption{This triangle cannot be a boundary because it does not have the augmentation property. }
\label{fig:trianglecycle}
\end{figure}

We show that this is the only obstruction to filling boundaries in $C_{\bullet}$: any cycle with the augmentation condition is a boundary. The idea is that non-trivial homology in $C_{\bullet}$ comes from relations among edge directions (see \Cref{sec:homology}, or \Cref{sec:koszul} for a different perspective). For example, the triangle illustrated in \Cref{fig:trianglecycle} is a cycle in $C_1$ that is not a boundary due to the relation $t_1+t_2+t_3=0$, where $t_1, t_2$ and $t_3$ are the edge vectors of the triangle. 

We first illustrate our solution with the example of a regular hexagon, as shown in \Cref{fig:hexagonfilling}. Of the three edge directions $t_1, t_2$ and $t_3$, we choose the free generating set $t_1$ and $t_2$. The two edges with direction $t_3$ can be connected by parallelepipeds, as shown in the right part of the figure. After subtracting the boundary of these parallelepipeds, the remaining boundary (bottom part of the figure) has only $t_1$ and $t_2$ as edge directions. Since $t_1$ and $t_2$ are independent, this remaining boundary can be filled.

\begin{figure}[h]
\centering
\includegraphics[page=2]{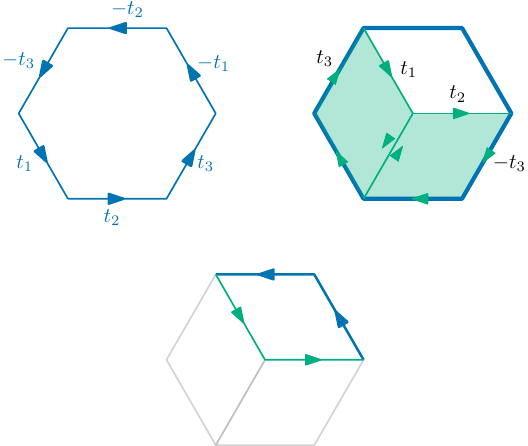}
\caption{ The hexagon has the augmentation property, so it can be filled. }
\label{fig:hexagonfilling}
\end{figure}

More generally, choose a set $E$ of edge directions that have no relations among them. 
Let $x \in C_k$ be a cycle with the augmentation property. 
Write $X=x_0+\dots+x_k$, where $x_i$ is a parallelepiped with $i$ edge directions in $E$ and $k-i$ edge directions not in $E$. Letting $x_j$ be the first non-zero term, we can find a boundary $y=y_j+\dots+y_k$ where $y_i$ has $i$ directions in $E$ such that $y_j=x_j$. The $j=0$ case requires the augmentation condition. By subtracting $y$ from $x$, we reduce the number of non-zero terms. Inductively, this implies that $x$ is a boundary.

We are still not done.
By induction, we constructed $x$ such that $\rho(x)$ is $\Delta_n$.
To perform the induction step,
we take $x'$ to be copies of $x$ for the facets of $\Delta_{n+1}$ so that $\rho(x') = d(\Delta_{n+1})$.
We now want to fill $x'$, but in general it is not true that $x'$ satisfies the augmentation condition. 


The last ingredient is the doubling map $S: C_{\bullet}\rightarrow C_{\bullet}$.
It takes a $k$-dimensional parallelepiped, scales it by a factor of $2$, and then subdivides the scaled parallelepiped into $2^k$ copies of the original one. This is illustrated in \Cref{fig:doublingmap}.
The element $x''=2^k x'-S x'$ is a cycle and it satisfies the augmentation property, so it can be filled. 
Geometrically, $S$ is just a dilation by a factor of two:
\[\rho(S(x'))=2\rho(x').\] 
So, for $n-1\ge 2$,
\[\rho\left(\frac{x''}{2^{n-1}-2}\right)=\rho(x')
= d(\Delta_{n+1}).\] 
Because $x''$ can be filled, there is $y$ such
that \[d(y) = \frac{x''}{2^{n-1}-2}.\]
Finally, by injectivity, the induction is complete
\[\rho(y) = \Delta_{n+1}.\]

\begin{figure}[h]
\centering
\includegraphics[width=7cm,page=3]{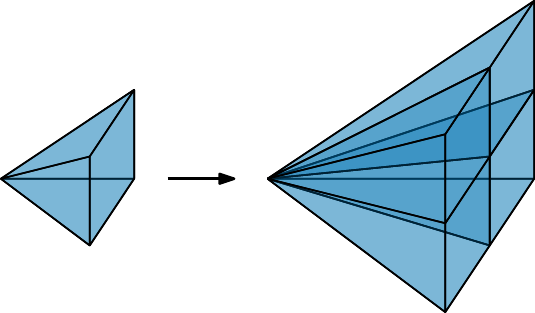}
\caption{ The doubling map $S$ scales a parallelogram in $C_k$ by a factor of $2$, then subdivides the parallelepiped face into $2^k$ pieces, which are each a translated copy of the original parallelepiped face. In $P_{\bullet}$, this corresponds to scaling a paramid by $2$ and then subdividing into $2^k$ pieces. But subdivision does not change a class in $P_{\bullet}$, so $\rho\circ S$ is equal to $2\rho$. }
\label{fig:doublingmap}
\end{figure}

\subsection{Overview of the Direct Implementation}
The second proof we present is based on directly constructing a family of paramids, whose indicator functions exhibit a combinatorial cancellation when being summed up. 
Here the goal is to build the regular $(n-1)$-simplex $\Delta = \conv \rset{e_i}{i \in \upto{n}}$ that sits inside the space $\rset{x \in \R^n}{\sum_i x_i = 1}$.
There are two main ingredients to the construction.
The first step is a comparably simple subdivision step, and the second step is where the non-trivial cancellation happens.
By the valuation property (\cref{lem:valuation}), it suffices to establish suitable identities up to lower-dimensional terms.

The first step is a geometric identity obtained by subdividing the simplex. Fix a nonempty subset $R \subsetneq \upto{n}$. For $v \notin R$, consider the simplex
\[
Q_{v,R}
\coloneqq (n - \abs{R})\Delta_{-v} - \sum_{i \notin  R \cup \set{v}} e_i ,
\]
where $\Delta_{-v} = \conv\rset{e_i}{i \neq v}$ is the facet opposite to $e_v$.
The simplex $Q_{v,R}$ is a dilation and translation of the facet $\Delta_{-v}$.
Cone it with $e_v$ to get the simplex
\[
    \pyr{e_v}{Q_{v,R}}
    \coloneqq
    \conv(\set{e_v} \cup Q_{v,R}).
\]
See \Cref{fig:simPar} for an illustration. The simplices $\pyr{e_v}{Q_{v,R}}$ for $v \notin R$ partition the simplex $\Delta$ up to dilation and a translation, so the valuation property of support functions implies
\[
  \MAX_n = h_\Delta \equiv \frac{1}{n - |R|} \sum_{v\notin R}
  h_{\operatorname{pyr}_{e_v}(Q_{v,R})} ,
\]
where $\equiv$ means up to support functions of lower-dimensional pieces.

\begin{figure}
    \centering
    \includegraphics{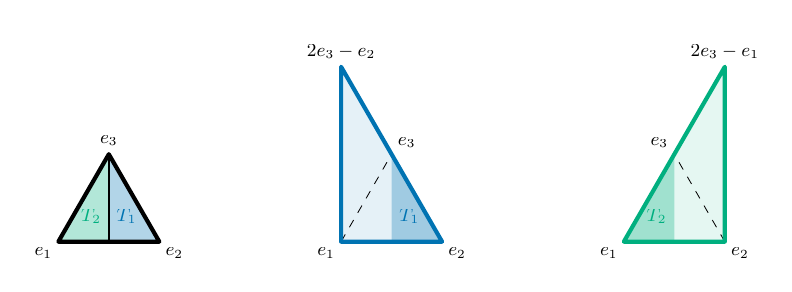}
    \caption{On left, a partition of the simplex $\Delta$ to simplices  with $R = \set{3}$. On middle and right, the two pyramids $\pyr{e_v}{Q_{v,R}}$ for $v=1$ and $v=2$.}
    \label{fig:simPar}
\end{figure}

This gives a family of identities expressing the support function of the original simplex in terms of other simplices. These identities alone do not yield the two-layer representations because the pyramids $\pyr{e_v}{Q_{v, R}}$ are simplices rather than paramids.

The second weightier step replaces the simplices $Q_{v,R}$ by parallelepipeds, which we cone with the apex $e_v$ to paramids. 
For each vertex $v$, we construct a family of parallelepipeds $B_{v,\pi,M}$ indexed by permutations $\pi$ and sets of positions $M$. Their signed indicator sum can be evaluated by viewing the input-coordinates as steps of a walk from~$0$ to~$1$. Inclusion--exclusion cancels all walks except the \emph{thermostat} walks which rise below level~$1$ and fall above level~$1$.
Their signed count turns out to be exactly $(-1)^{s-1}s!(n-1-s)!$, where $s$ is the number of positive coordinates of the (generic) input. 
Binomial inversion expresses this count as a linear combination of the indicator sums of the auxiliary simplices $Q_{v,R}$:
\begin{equation*}
\sum_{\pi}\sum_M (-1)^{\abs{M}} \ind{B_{v,\pi,M}}
\equiv
(n-1)!\sum_{R}
\frac{(-1)^{\abs{R} - 1} \abs{R}}{n - \abs{R}}
\ind{Q_{v,R}}.
\end{equation*}
This identity replaces the simplices $Q_{v,R}$ by parallelepipeds $B_{v,\pi,M}$. See \Cref{fig:simpAndPar} for illustrations.


\begin{figure}
  \centering
  \resizebox{\linewidth}{!}{%
    \input{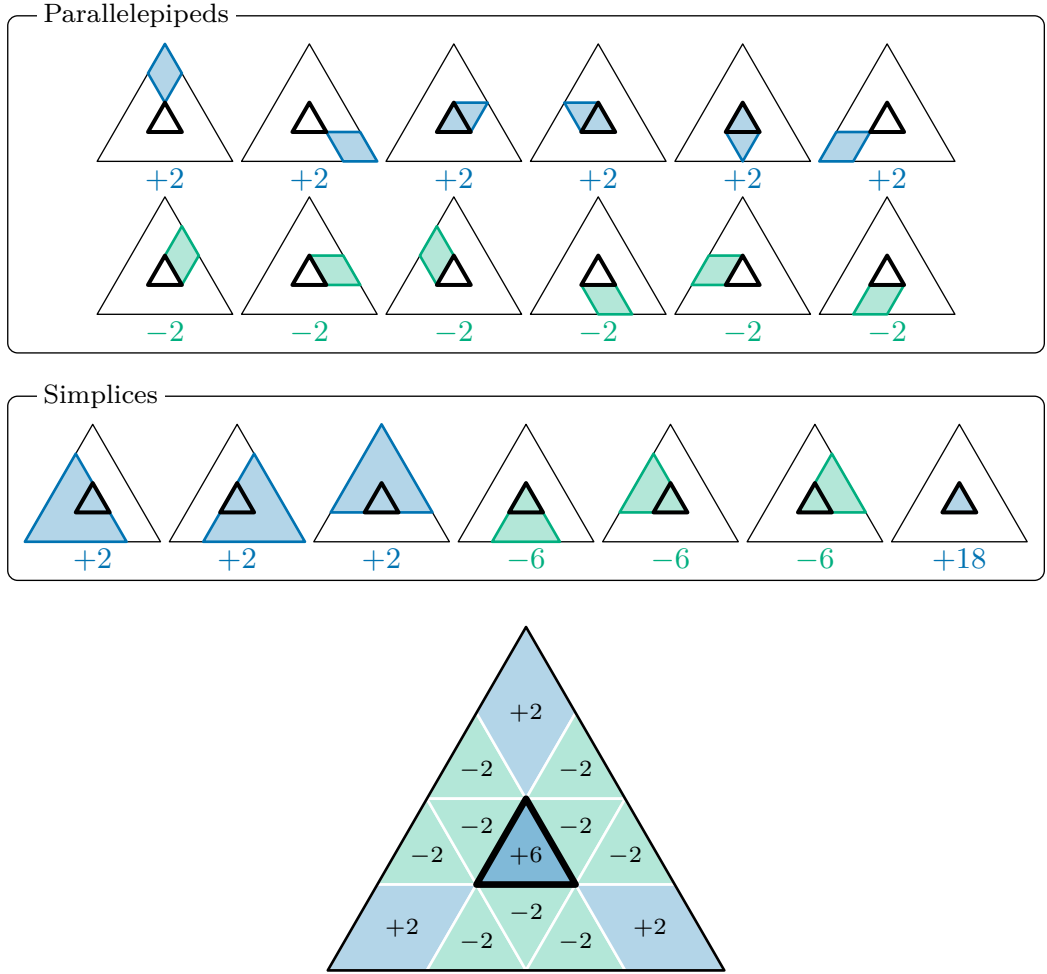}%
  }
  \caption{At the top, the parallelepipeds $B_{v,\pi,M}$ and the simplices $Q_{v,R}$ together with their respective coefficients. At the bottom, the total sum is the same for the parallelepipeds and for the simplices.}
  \label{fig:simpAndPar}
\end{figure}

But this identity by itself is not useful, because parallelepipeds are not enough.
We need to cone the identity. 
Coning by $e_v$  brings the needed paramids to the left-hand side.

Summing over all vertices $v$ and applying the geometric identities from the first step gives $h_\Delta$ up to a lower-dimensional remainder term $R_{n-1}$:
\[
    \MAX_n = h_\Delta =
    \frac{(-1)^n}{ n!} \sum_{v} \sum_{\pi} \sum_{M} (-1)^{\abs{M}}
    h_{\pyr{e_v}{B_{v,\pi,M}}} + R_{n-1}.
\]
This gives a two-layer representation of $\maxfn{n}$, because all paramids have support functions representable in two layers and the lower dimensional correction $R_{n-1}$ is handled by induction.

\section*{ Statement of AI Use }

The authors verified the correctness and originality of all content, including references. 

\subsection*{Homological Algebra Approach }

The proof contains significant technical contributions from both AI models and the authors. 
We used ChatGPT Sol 5.6 and ChatGPT Astra 6. 

The major contribution of the human authors is the overall proof strategy. The constructions of the chain complexes $P_{\bullet}$ and $\Lambda_{\bullet}$, the idea to represent elements of $\Lambda_{\bullet}$ as images of a coning map $\rho:C_{\bullet}\rightarrow P_{\bullet}$, and the injectivity of the top boundary (\Cref{lem:topinjectivity}) are all authors' ideas. AI was used throughout the process to check proofs and look for simplifications.

The largest contributions from AI are the proofs of \Cref{lem:aug} and the doubling map used to prove \Cref{lem:cyclecorrection}. The proof of \Cref{lem:lowerdimcharacterization} that appears in this paper is also due largely to AI, but another proof was previously known to the authors (see \Cref{sec:flaglowerdim}). 


\subsection*{Direct Combinatorial Approach }
We used Claude Opus 5.5 extensively in developing the mathematical content of the combinatorial proof.

The outline of the proof was conceptualized by the authors. Our idea was to generalize the approach of subdividing into pyramids over zonotopes used in \cite{bakaevBetterNeuralNetwork2026} to prove $\maxfn{5} \in \relu_{5,2}$, handling lower-dimensional intersections by induction. We proposed combining this approach with ideas from the computational proof that $\maxfn{12} \in \relu_{12,2}$ in \cite{ruess2026shallower}. Working interactively with the authors, the model implemented and carried out calculations in small dimensions, identified patterns in the resulting solutions and proposed the concrete constructions underlying \cref{lem:thermostatN,lem:replacement}.

LLMs, primarily ChatGPT Sol 6.1, were also used to produce some of the TikZ figures.

\section*{Acknowledgments}

The authors thank the members of the Neural Polytopes Zulip community.
F.B. and A.Y. are supported by the DNRF Chair grant
20231101-28697. J.S. is supported by Starting Grant 1054-00032B from the Independent Research Fund Denmark under the Sapere Aude research career programme, and by the Carlsberg Foundation, grant CF24-1929.



\part{A Homological Algebra Proof}

\section{ The Chain Complex of Polytopes }\label{sec:minkowskicomplex}

We start by defining a chain complex $P_{\bullet}^n$. Here $P_k^n$ should represent ``$k$-dimensional polytopes in $\mathbb{Q}^n$ under Minkowski sum, modulo polytopes of dimension $<k$''.

For a finite-dimensional real vector space $S$, let $V_S$ be the vector space of formal Minkowski-linear combinations of polytopes. Formally, $V_S$ is the $\mathbb{Q}$-vector space spanned by pairs $(X, Y)$ of polytopes, modulo the relation $(X, X)=0$. Addition in $V_S$ is given by Minkowski sum $(X_1, Y_1)+(X_2, Y_2)=(X_1+X_2, Y_1+Y_2)$. Multiplication by a constant is given by $\lambda(X, Y)=(\lambda X, \lambda Y)$ when $\lambda$ is positive, and $\lambda(X, Y)=(|\lambda| Y, |\lambda| X)$ when $\lambda$ is negative. 

If $S'$ is a subspace of $S$, then there is a natural inclusion $V_{S'}\rightarrow V_S$. We define $\widetilde{V}_S$ to be the quotient of $V_S$ by the sum over all proper subspaces $S'$ of $S$ of $V_{S'}$. 

Formally, for $k\ge 1$, we define $P_{k}^n$ to be the direct sum of $\widetilde{V}_S$ over each oriented $k$-dimensional subspace $S$ of $\mathbb{Q}^n$, with relations that reversing the orientation of a subspace corresponds to reversing sign in $V_S$. 

$P_0^n$ will be mostly irrelevant for our calculations, but we define it here for completeness. For $k=0$, $P_{0}^n$ is defined as a direct sum over $\mathbb{Q}^n$-indexed copies of $\mathbb{Q}$, modulo global translations. Note a translation induces a permutation of the summands. Here we are adding relations $x-t(x)=0$, where $x\in \bigoplus_{\mathbb{Q}^n}\mathbb{Q}$ and $t$ is the linear map induced by the translation $t$ permuting the factors. 

We want to define the boundary $d:P_{k+1}^n\rightarrow P_{k}^n$. We start with $k\ge 1$. Now $P_{k+1}$ is generated (as a $\mathbb{Q}$-vector space) by classes of form $(X, 0)$ with $X$ a polytope in a $k+1$-dimensional oriented subspace $S\subseteq \mathbb{Q}^n$. The boundary of $(X, 0)$ is defined as the sum over all unit vectors $v\in S$ of the face of $X$ in direction $v$, translated so that it lies in the subspace $U$ orthogonal to $v$. The orientation $A$ of $U$ is chosen so that $v\wedge A$ represents the orientation of $S$. Note that this term is zero in $P_k^n$ unless $v$ is orthogonal to a $k$-dimensional face of $X$, so only finitely many terms are non-zero.

When $k=1$, $P_1^{n}$ is generated by classes of form $(X, 0)$, with $X$ a line segment in a $1$-dimensional space $S$. Let $a$ and $b$ be the endpoints of $X$, chosen so that the vector $a-b$ induces the orientation on $S$. We set $d((X, 0))$ to be the class of $e_a-e_b$, where $e_a$ is the standard basis vector corresponding the $a$-component of $\bigoplus_{\mathbb{Q}^n}\mathbb{Q}$, and $e_b$ is the one of the $b$ component.

We need to check that $d$ is well-defined and linear. Linearity is straightforward: if $X$ and $Y$ are polytopes, then the face of $X+Y$ in direction $v$ is the Minkowski sum of the faces of $X$ and $Y$ in direction $v$. 

It is clear that $d((X, X))=0$. If $X$ is $\le k$-dimensional and $S$ is a $k+1$-dimensional subspace containing $X$, then $(X, 0)$ is $0$ in $\widetilde{V}_S$, so we also need to check that $d((X, 0))=0$. If $X$ is $<k$ dimensional, then it has no $k$-dimensional faces, so $d((X, 0))$ is $0$ in $P_{k}^n$. If $X$ is exactly $k$ dimensional, then it has two identical $k$-dimensional faces in the two directions orthogonal to $X$. These induce opposite orientations of the same class in $P_{k}^n$, so they cancel, and again $d(X)=0$. This shows that $d$ is well-defined.

We need to check that $d^2=0$. Let $X$ be an $n$-dimensional polytope. Every $n-2$-dimensional face $F$ of $X$ is a face of two $n-1$ dimensional faces $Y_1$ and $Y_2$ of $X$. When we take $d^2(x)$, a class corresponding to $F$ appears twice, but with opposite orientations. So $d^2((X, 0))=0$, and $d^2=0$ by linearity. 

We define $\Lambda^n_{\bullet}$ to be the subcomplex of $P_{\bullet}^n$ where $\Lambda_k^n$ is generated by classes of paramids of dimension $k$. For $k\le 2$, a paramid is a simplex, so $\Lambda^n_{\bullet}=P_{\bullet}^n$. For $k\ge 3$, the parallelepiped face of a paramid is $0$ in $P_{k-1}^n$, and all the other faces of a paramid are paramids of one dimension lower. So $\Lambda_{\bullet}^n$ is closed under the boundary.

Paramids can be represented in depth $2$. If the $n-1$-simplex can be represented in depth $2$, then anything of dimension $<n$ can be represented in depth $2$, so everything in $\Lambda_{n}^n$ can be represented in depth $2$. Our goal is to show that the class of the $n$-simplex is in $\Lambda_n^n$ for each $n$. By induction, this will show that the $n$-simplex is in depth $2$ for all $n$.

The key property of $P_{\bullet}^n$ is the following lemma, which says that a polytope is uniquely defined by its boundary, even modulo lower-dimensional terms:

\begin{lemma}\label{lem:topinjectivity}
The top differential $d:P_{n}^n\rightarrow P_{n-1}^n$ is injective.
\end{lemma}

In order to prove \Cref{lem:topinjectivity}, we need to understand the space $\widetilde{V}$ better. Our goal is to prove the following characterization of $\widetilde{V}_S$:

\begin{lemma}\label{lem:lowerdimcharacterization}
Let $X$ and $Y$ be polytopes in $\mathbb{Q}^n$, with $n\ge 2$. Then $(X, Y)=0$ in $\widetilde{V}_{\mathbb{Q}^n}$ if and only if, for each non-zero direction $v$, the following property holds: let $F_1$ and $F_2$ be the faces of $X$ in directions $v$ and $-v$ respectively, and let $G_1$ and $G_2$ be the faces of $Y$ in directions $v$ and $-v$ respectively. Then $(F_1+G_2, F_2+G_1)=0$ in $\widetilde{V}_{S}$, where $S$ is the subspace of $\mathbb{Q}^n$ orthogonal to $v$, and we translate $F_1, F_2, G_1$ and $G_2$ so that they lie in $S$.
\end{lemma}

The ``only if'' direction appears in \cite{KMPS2025}. We reproduce the proof below:

\begin{proof}[Proof of the ``only if'' direction.]
The condition $(X, Y)=0$ in $\widetilde{V}_{\mathbb{Q}^n}$ is equivalent to saying that there exist polytopes $A_1, \dots, A_k$ and $B_1, \dots, B_j$, each of dimension $<n$, such that $X+A_1+\dots+A_k=Y+B_1+\dots+B_j$.

First, suppose that such $A_i$ and $B_i$ exist. Fix a non-zero direction $v$. For each $i$, let $A_i'$ be the face of $A_i$ in direction $v$, and let $A_i''$ be the face of $A_i$ in direction $-v$, and define $B_i'$ and $B_i''$ similarly. Let $F_1$ and $F_2$ be the faces of $X$ in directions $v$ and $-v$ respectively, and $G_1$ and $G_2$ similarly for $Y$. Now $F_1+A_1'+\dots+A_k'=G_1+B_1'+\dots+B_j'$ and $F_2+A_1''+\dots+A_k''=G_2+B_1''+\dots+B_j''$, so:

\[F_1+G_2+A_1'+\dots+A_k'+B_1''+\dots+B_j''=F_2+G_1+A_1''+\dots+A_k''+B_1'+\dots+B_j'\]

For each $A_i$, we have two cases: either $v$ is not orthogonal to $A_i$, in which case $A_i'$ and $A_i''$ have dimension $<n-1$, or $v$ is orthogonal to $A_i$, in which case $A_i'=A_i''$. In the latter case, we can cancel $A_i'$ and $A_i''$ from both sides of the equation. A similar claim holds for $B_i'$ and $B_i''$. After canceling equivalent terms, we see that $(F_1+G_2, F_2+G_1)=0$ in $\widetilde{V}_S$, where $S$ is the subspace orthogonal to $v$. The assumption that $n\ge 2$ means that we can translate $F_1+G_2$ and $F_2+G_1$ so that they lie in $S$. 
\end{proof}

The other direction of the lemma could be deduced from the results in \cite{KPS2026} (see \Cref{sec:flaglowerdim}). Instead, we will use a result of Jessen and Thorup \cite{JT78}. They studied an algebra of polytopes modulo translations and a subdivision relation. Their results are stated for polytopes with vertices in an arbitrary subfield of $\mathbb{R}$. We specialize to the case of vertices in $\mathbb{Q}$.

For a finite-dimensional rational vector space $S$, we define $\mathcal{B}_S$ to the abelian group freely generated by generators corresponding to full-dimensional polytopes in $S$. There is a homomorphism $\text{Vol}_S:\mathcal{B}_S\rightarrow \mathbb{Q}$ that sends a polytope to its volume.

We define $\Pi_S$ to be $\mathcal{B}_S$ modulo translations and subdivision relations that say that $X_1+X_2=Z$ if $X_1, X_2$ and $Z$ are polytopes where $X_1$ and $X_2$ are interior-disjoint and $X_1\cup X_2=Z$. It is straightforward to see that $\text{Vol}_S$ descends to a map $\Pi_S\rightarrow \mathbb{Q}$.

Let $Q$ be a codimension $1$ subspace of $S$ and let $v$ be a vector not in $Q$. Then we define a map $h_{Q, v}: \mathcal{B}_S\rightarrow \mathcal{B}_Q$ as follows: for a polytope $X$ in $S$, let $F_1$ and $F_2$ be the faces of $X$ in the directions orthogonal to $Q$, with the normal to $F_1$ being in the $+v$ direction and the normal of $F_2$ being in the $-v$ direction, and translate $F_1$ and $F_2$ by multiples of $v$ so that they are in $Q$. We view $F_1$ or $F_2$ as being $0$ if it is not full-dimensional.

Similar to the proof of the ``only if'' direction of \Cref{lem:lowerdimcharacterization}, it is not too hard to see that, if $x\in \mathcal{B}_S$ is $0$ in $\Pi_S$, then $h_{Q, v}(x)=0$ in $\Pi_Q$. Jessen and Thorup prove the following:

\begin{theorem}(\cite{JT78}, Theorem 1)\label{thm:pivanishingcondition}
An element $x\in \mathcal{B}_S$ is $0$ in $\Pi_S$ if and only if $\text{Vol}_S(x)=0$ and, for every codimension $1$ subspace $Q\subseteq S$ and $v\notin Q$, $h_{Q,v}(x)$ is $0$ in $\Pi_Q$.
\end{theorem}

We define a map $\xi:\mathcal{B}_S\rightarrow V_S$ that sends a polytope $X$ to the class of $(X, 0)$. Clearly, $\xi$ is surjective. When $n\ge 1$, \Cref{lem:additivity} says that $\xi$ descends to a map $\Pi_S\rightarrow \widetilde{V}_S$. In order to use \Cref{thm:pivanishingcondition}, we need to understand $\xi$ better.

For $k\in \mathbb{N}_+$, the dilation $\lambda_k:\Pi_S\rightarrow \Pi_S$ is the map that sends a polytope $X$ to $kX$. An element $x\in \Pi_S$ is said to have degree $i$ if $\lambda_k(x)=k^ix$ for each $k$. We let $\Pi_S^i$ be the subgroup of elements of degree $i$. Jessen and Thorup show the following:

\begin{theorem}(\cite{JT78}, Section 5)\label{thm:pispectraldecomp}
The group $\Pi_S$ is the direct sum of $\Pi_S^i, \dots, \Pi_S^{\dim(S)}$.
\end{theorem}

An example with $S=\mathbb{Q}^2$ is illustrated in \Cref{fig:dilationeigenspaces}. In this case, $\Pi_{\mathbb{Q}^2}^2$ is spanned by parallelograms, while $\Pi_{\mathbb{Q}^2}^1$ is spanned by terms that are the difference of a triangle and a parallelogram with the same area.

\begin{figure}
\centering
\includegraphics[width=13cm,page=1]{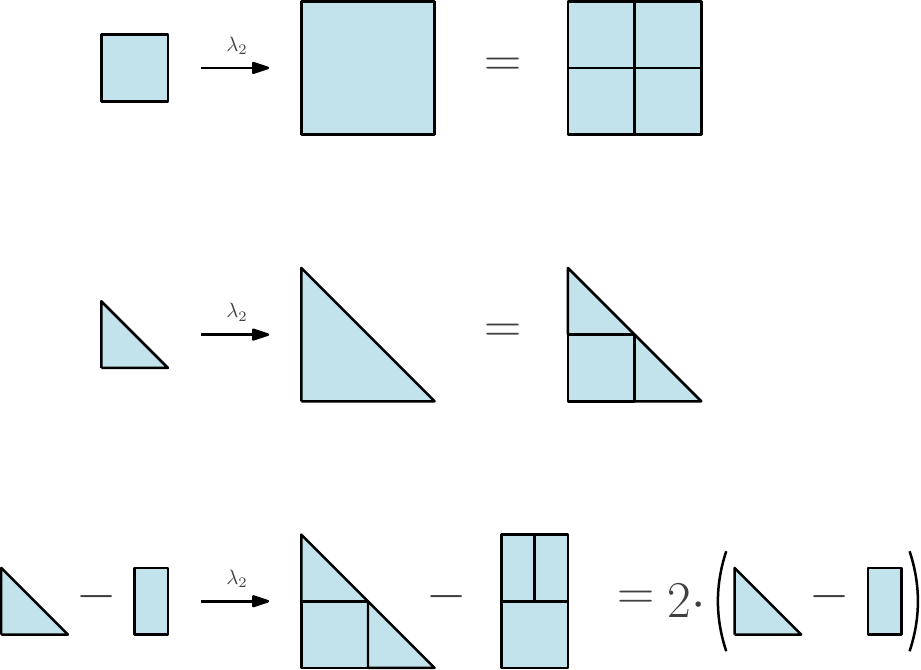}
\caption{ Top: the dilation of a parallelogram by $2$ can be subdivided into $4$ copies of the original parallelogram. Middle: the dilation of a triangle by $2$ can be subdivided into two copies of the original triangle and one parallelogram. So the difference of a triangle and a rectangle of the same area is in $\Pi_{\mathbb{Q}^2}^1$ (bottom). }
\label{fig:dilationeigenspaces}
\end{figure}

Our next step is to show that $\xi$ gives an isomorphism $\Pi_S^1\rightarrow \widetilde{V}_S$. Note that an element of $\Pi_S^i$ for $i<\text{dim}(S)$ has volume $0$. Indeed, $\text{Vol}_S(\lambda_k(x))=k^{\text{dim}(S)}\text{Vol}_S(x)$. We also note that every $h_{Q, v}$ map commutes with every $\lambda_k$, so $h_{Q, v}$ preserves the decomposition into $\Pi_S^{i}$.

If $S$ has dimension $\ge 2$ and $Q\subseteq S$, then we can define a unique map $h_{Q, v}:V_S\rightarrow V_Q$ similarly to $h_{Q, v}$ on $\mathcal{B}_S$. Note that $h_{Q, v}\circ\xi=\xi \circ h_{Q,v}$. We can now determine the kernel of $\xi$:

\begin{lemma}\label{lem:xikernel}
The kernel of $\xi: \Pi_S\rightarrow \widetilde{V}_S$ is the sum of $\Pi_S^i$ for $i>1$.
\end{lemma}

\begin{proof}
First, we notice that every degree $>1$ component of $\Pi_S$ is in the kernel of $\xi$. This is because the dilation $\lambda_k$ commutes with multiplication by $k$ in $\widetilde{V}_S$. That is, $\xi(\lambda_k(x))=k\xi(x)$.

So it remains to show that this is the entire kernel, i.e. that $\xi$ is injective on $\Pi_S^1$. This is by induction on $\text{dim}(S)$. The base case is $\text{dim}(S)=1$. In this case, $\Pi_S=\Pi_S^1$ and it is easy to see that $\xi:\Pi_S\rightarrow \widetilde{V}_S$ is an isomorphism.

Now let $\text{dim}(S)\ge 2$. Let $x\in \Pi_S^1$ with $\xi(x)=0$. Represent $x$ as a sum of polytopes $A_1+\dots+A_k-B_1-\dots-B_j$. So $\xi(x)=(A_1+\dots+A_k, B_1+\dots+B_k)=0$ in $\widetilde{V}_S$. 

The ``only if'' part of \Cref{lem:lowerdimcharacterization} exactly says that $h_{Q, v}((A_1+\dots+A_k, B_1+\dots+B_k))=0$ for all $Q, v$. So $\xi(h_{Q, v}(x))=h_{Q,v}(\xi(x))=0$ for all $Q, v$. 

Now $h_{Q,v}(x)$ is in $\Pi_Q^1$ for all $Q, v$, so by inductive assumption, $\xi(h_{Q,v}(x))=0$ implies $h_{Q,v}(x)=0$. We also know that $\text{Vol}_S(x)=0$. So by \Cref{thm:pivanishingcondition}, $x=0$ in $\Pi_S$. By induction, this proves the claim. 
\end{proof}

We know that $\xi$ is surjective, so \Cref{lem:xikernel} shows that $\xi: \Pi_S^1\rightarrow \widetilde{V}_S$ is an isomorphism. We can now complete the proof of \Cref{lem:lowerdimcharacterization}:

\begin{proof}[Proof of the ``if'' part of \Cref{lem:lowerdimcharacterization}]
The conditions of the lemma say exactly that $h_{Q, v}((X, Y))=0$ for every codimension $1$ subspace $Q\subseteq \mathbb{Q}^n$ and $v\notin Q$. Since $\xi$ is an isomorphism $\Pi_S^1\rightarrow \widetilde{V}_S$, there is some $x\in \Pi_S^1$ such that $\xi(x)=(X, Y)$. For every $Q, v$, we have $\xi(h_{Q, v}(x))=h_{Q, v}(\xi(x))=h_{Q, v}((X, Y))=0$. Since $h_{Q, v}(x)\in \Pi_S^1$, we have $h_{Q, v}(x)=0$ by \Cref{lem:xikernel}. Since $n\ge 2$, $\text{Vol}_{\mathbb{Q}^n}(x)=0$, so $x=0$ in $\Pi_S$. So $(X, Y)=\xi(0)=0$ in $\widetilde{V}_S$.
\end{proof}

\Cref{lem:topinjectivity} follows immediately from \Cref{lem:lowerdimcharacterization}.

\section{ The Chain Complex of Parallelepipeds }\label{sec:parallelepipeds}

Let $T=\{t_1, \dots, t_r\}$ be the set of edge directions of the $n$ simplex (with an arbitrary direction chosen for each). So $r={n+1\choose 2}$. Let $\Gamma=\mathrm{Span}_\mathbb Z\{t_1, \ldots, t_r\}$ be the lattice in $\mathbb{Q}^n$ generated by the $t_i$. 

The ring $\mathbb{Q}[\Gamma]$ consists of formal $\mathbb{Q}$-linear combinations of elements of $\Gamma$, where multiplication is given by the group law on $\Gamma$. We write $g_1, \dots, g_r$ for the generators of $\mathbb{Q}[\Gamma]$ corresponding to $t_1, \dots, t_r$. So, for example, $g_ig_j$ is the element of $\mathbb{Q}[\Gamma]$ corresponding to $t_i+t_j$. The element $g_i+g_j$ is viewed formally as a superposition of $t_i$ and $t_j$. 

We define a chain complex $C_{\bullet}^n$, which consists of formal oriented parallelepipeds with edge directions in $T$ and vertices in $\Gamma$, with the differential defined in the standard way. 

Let $E=\mathbb{Q}^r$ with basis $e_1, \dots, e_r$. We can view $C_k^n$ as the $\mathbb{Q}$-module $\mathbb{Q}[\Gamma]\otimes_{\mathbb{Q}}\bigwedge^{k}_{\mathbb{Q}}E$. This is generated by elements of the form $\gamma \otimes e_{i_1}\wedge\dots\wedge e_{i_k}$ ($\gamma\in \Gamma$), which we view as corresponding to the parallelepiped:

\[\gamma+[0, t_{i_1}]+\dots+[0, t_{i_k}]\]

The boundary map on $C_k^n$ is then defined as the standard cubical boundary. We have $d(e_i)=g_i-1$, this extends to higher orders using the graded Leibniz rule:

\[d(\gamma\otimes e_{i_1}\wedge \dots \wedge e_{i_k})=\sum_{j=1}^{k}(-1)^{j-1}\gamma(g_{i_j}-1) \otimes x_{i_1}\wedge \dots \wedge \widehat{x_{i_{j}}}\wedge \dots\wedge x_k\]

The $\rho:C_{k}^{n}\rightarrow P_{k+1}^n$ is defined as follows: let $x=\gamma\otimes e_{i_1}\wedge \dots \wedge e_{i_k}$ be a parallelepiped in $C_k^n$. Let $X$ be the convex hull of $0$ and the vertices of $x$, which is a polytope in $S=\text{span}(\gamma, t_{i_1}, \dots, t_{i_k})$. If $\gamma, t_{i}, \dots, t_{i_k}$ are linearly independent, then the cone is non-degenerate and $\rho(x)$ is the class corresponding to $X$ in $S$, with the orientation of $S$ given by $t_{i_1}\wedge\dots\wedge t_{i_k}\wedge \gamma$. Otherwise, the cone is degenerate and $\rho(x)=0$. 

\begin{lemma}
The map $\rho$ is well-defined.
\end{lemma}

\begin{proof}
The only thing to check is that $\rho$ respects the relation:

\[\gamma\otimes e_{i_1}\wedge\dots \wedge e_{i_j}\wedge \dots\wedge e_{i_{\ell}}\wedge \dots\wedge e_{i_k}=-\gamma\otimes e_{i_1}\wedge\dots \wedge e_{i_\ell}\wedge \dots\wedge e_{i_{j}}\wedge \dots\wedge e_{i_k}\]

\noindent in $C_k^n$. This is straightforward. 
\end{proof}

We now claim that $\rho$ commutes with the differentials on $C_{\bullet}^n$ and $P_{\bullet}^n$ in degrees at least $2$:

\begin{lemma}\label{lem:rhodcommutes}
If $x\in C_{k}^n$ for $k\ge 2$, then $d(\rho(x))=\rho(d(x))$.
\end{lemma}

\begin{proof}
It is sufficient to check the claim for $x=\gamma\otimes x_{i_1}\wedge \dots\wedge x_{i_k}$ a parallelepiped. In this case, we have:

\[d(x)=\sum_{j=1}^{k}(-1)^{j-1}\gamma (g_{i_j}-1)\otimes x_{i_1}\wedge \dots \wedge \widehat{x_{i_{j}}}\wedge \dots\wedge x_k\]

There are two cases to check, depending on whether $x$ is degenerate. First, we consider the case that $x$ is not degenerate. In this case, $\rho(x)$ is the class in $P_{k+1}^n$ of the polytope $X$, which is the convex hull of $0$ and $\gamma+[0, t_{i_1}]+\dots+[0,t_{i_k}]$. The orientation is given by $t_{i_1}\wedge\dots\wedge t_{i_k}\wedge \gamma$

$X$ has $2k+1$ facets. The face not containing $0$ is $\gamma+[0, t_{i_1}]+\dots+[0,t_{i_k}]$. This is a Minkowski sum of intervals, which have dimension $1<k$. So the contribution of this face to $d(\rho(x))$ is $0$.

For each $j\in \{1, \dots, k\}$, we have two faces $F_j$ and $G_j$ of $X$. $F_j$ the convex hull of $0$ and $\gamma+\sum_{\ell\ne j}[0, t_{i_\ell}]$ and $G_j$ is the convex hull of $0$ and $\gamma+t_{i_j}+\sum_{\ell\ne j}[0, t_{i_\ell}]$. 

Now $t_{i_j}$ points out from $G_j$ and $-t_{i_j}$ points out from $F_j$. So the orientation $A$ of $F_j$ is chosen so that $(-t_{i_j})\wedge A$ is the same as the orientation of $t_{i_1}\wedge\dots\wedge t_{i_k}\wedge \gamma$ of $X$. So $(-1)^{j} t_{i_1}\wedge\dots\wedge t_{i_{j-1}}\wedge t_{i_{j+1}}\wedge\dots\wedge t_{i_k}\wedge \gamma$ is a representation of the orientation of $F_j$. Similarly, the orientation of $G_j$ is represented by $(-1)^{j+1}t_{i_1}\wedge\dots\wedge t_{i_{j-1}}\wedge t_{i_{j+1}}\wedge\dots\wedge t_{i_k}\wedge \gamma$, which is equal to $(-1)^{j+1}t_{i_1}\wedge\dots\wedge t_{i_{j-1}}\wedge t_{i_{j+1}}\wedge\dots\wedge t_{i_k}\wedge (\gamma+t_{i_j})$. This shows that contributions of $F_j$ and $G_j$ to $d(\rho(x))$ are exactly:

\[\rho\left((-1)^{j-1}\gamma (g_{i_j}-1)\otimes x_{i_1}\wedge \dots \wedge x_{i_{j-1}}\wedge x_{i_{j+1}}\wedge \dots\wedge x_k\right)\]

So $d(\rho(x))=\rho(d(x))$ when the cone of $x$ is non-degenerate. We now move on to the degenerate case. Similarly to before, let $X$ be the convex hull of $0$ and $\gamma+\sum_{j}[0, t_{i_j}]$. Since the cone is degenerate, $X$ is at most $k$ dimensional. If $X$ is $\le k-1$ dimensional, then every face in $d(x)$ is degenerate, so $\rho(d(x))=0$. The remaining case is when $X$ is exactly $k$ dimensional. Let $S$ be the $k$-dimensional subspace containing $X$, and choose an arbitrary orientation of $S$.

In this case, define $F_j$ and $G_j$ as oriented polytopes in the same way as before. We also define a parallelepiped face $Z=\gamma+\sum_{j}[0, t_{i_j}]$. Since $Z$ is zero in $P_{k}^n$, it is clear that $\rho(d(x))$ is the sum of classes corresponding to the $F_j$, $G_j$ and $Z$.

Now each of the $F_j$, $G_j$ or $Z$ is either degenerate, has positive orientation (with respect to the orientation chosen for $S$), or has negative orientation. We claim that the positive-orientation faces give a subdivision of $X$, and similarly the negative orientation faces also give a subdivision of $X$. This is straightforward. 

So by \Cref{lem:additivity}, the contributions from the positively-oriented pieces cancel with the contributions from the negatively-oriented pieces, so $\rho(d(x))=0$, as required.
\end{proof}

The image of $\rho$ in $P_{\bullet}^n$ lies in the subcomplex $\Lambda_{\bullet}^n$ generated by paramids. Note that parallelepipeds in $C_{\bullet}^n$ are restricted to having edge directions in $T$, so being in the image of $\rho$ is more restrictive than just being in $\Lambda_{\bullet}^n$. 

\subsection{ The Homology of $C_{\bullet}^n$ } \label{sec:homology}

We can actually give a nice characterization of the homology of $C_{\bullet}^n$. Let $\hat{\Gamma}=\mathbb{Z}^r$, generated by $\hat{t}_1, \dots, \hat{t}_r$, and let $\pi$ be the projection $\hat{\Gamma}\rightarrow \Gamma$ given by $\hat{t}_i\rightarrow t_i$. 

Let $\hat{C}_{\bullet}^n$ be the chain complex of parallelepipeds (in $\mathbb{R}^r$) with vertices in $\hat{\Gamma}$ and edge direction in $\{\hat{t}_1, \dots, \hat{t}_r\}$. Now $\hat{C}_{\bullet}^n$ is a cubicalization of $\mathbb{R}^r$, and $C_{\bullet}^n$ is obtained from $\hat{C}_{\bullet}^n$ by identifying parallelepipeds that differ by a translation in $\ker(\pi)$.

Now $\ker(\pi)$ is a subgroup of $\mathbb{Z}^r$, so by the structure theorem for finitely generated abelian groups, it is freely generated by elements $a_1, \dots, a_{r-n}$ of $\mathbb{Z}^r$. The $t_i$ are edge vectors of the simplex, and it is not too hard to find such a set of generators explicitly.

So $C_{\bullet}^n$ is a cubicalization of a topological space which is the quotient of $\mathbb{R}^r$ by the action of the group generated by $a_1, \dots, a_{r-n}$. This space is homeomorphic to a product of $\mathbb{R}^{n}$ and $r-n$ circles. So $H_{k}(C_{\bullet}^n)$ is a direct sum of ${r-n\choose k}$ copies of $\mathbb{Q}$.

We didn't need to use anything special about $t_1, \dots, t_r$ here. In general, let $T=\{t_1, \dots, t_r\}$ be a set of directions generating an abelian group $\Gamma$, and let $C_{\bullet}^{T}$ be the chain complex of parallelepipeds with edges directions in $T$ and vertices in $\Gamma$. Then the argument above shows that $C_{\bullet}^{T}$ is a cubicalization of a topological space which is the quotient of $\mathbb{R}^r$ by the set of relations among the $t_i$. In particular, if the $t_i$ \emph{freely} generate $\Gamma$, then $C_{\bullet}^{T}$ has the homology of a point. This observation will be useful for filling cycles in $C_{\bullet}^n$. 

This last observation is also a consequence of a more general idea about the Koszul resolution of a regular sequence, see \Cref{sec:koszul}.

\subsection{ Filling Boundaries in $C_{\bullet}^n$ }

Let $D_{\bullet}^n$ be $C_{\bullet}^n$ modulo translations. Formally, we can think of $D_{k}^n$ as being equal to $\bigwedge_{\mathbb{Q}}^k E$ (recall $E=\mathbb{Q}^r$, generated by $e_1, \dots, e_r$). Let $\varepsilon$ be the map $C_{\bullet}^n\rightarrow D_{\bullet}^n$ that forgets translation in $C_{\bullet}^n$. Using $C_{k}^n=\mathbb{Q}[\Gamma]\otimes \bigwedge_{\mathbb{Q}}^k E$, $\varepsilon$ can be defined using the map $\mathbb{Q}[\Gamma]\rightarrow \mathbb{Q}$ that sends each $\gamma\in \Gamma$ to $1$. 

We note that every boundary $x$ in $C_{k}^n$ has $\varepsilon(x)=0$. Indeed, opposite pairs of faces of a parallelepiped cancel after taking $\varepsilon$. So this gives a necessary condition for a cycle $x$ to be a boundary. We claim that it is sufficient:

\begin{lemma}\label{lem:aug}
A cycle $x\in C_{k}^n$ is a boundary if and only if $\varepsilon(x)=0$
\end{lemma}

\begin{proof}
The ``only if'' direction is clear. So suppose that $x$ is a cycle with $\varepsilon(x)=0$. We need to show that $x$ is a boundary.

Without loss of generality, suppose that $t_1, \dots, t_n$ freely generate the rank-$n$ free abelian group $\Gamma$. For example, we can take $t_1, \dots, t_n$ to be the set of edges of the simplex containing one distinguished vertex. Write $A=\{t_1, \dots, t_n\}$ and $B=\{t_{n+1}, \dots, t_r\}$. Let $K_{\bullet}$ be the chain complex of parallelepipeds with edge directions in $A$ and vertices in $\Gamma$, and similarly let $T_{\bullet}$ be the chain complex of parallelepipeds with edge direction in $B$ and vertices in $\Gamma$. Write $\delta$ for the differential on $K_{\bullet}$ and $\partial$ for the differential on $T_{\bullet}$.

We claim that $C_{\bullet}^n\cong K_{\bullet}\otimes T_{\bullet}$. Formally, this says:

\[C_{k}^n\cong \bigoplus_{p+q=k}K_p\otimes_{\mathbb{Q}[\Gamma]} T_{q}\]

In other words, every parallelepiped in $C_k^n$ can be written as a formal product of one in $K_{\bullet}$ and one in $T_{\bullet}$. The claim is clear from $K_{p}=\mathbb{Q}[\Gamma]\otimes \bigwedge_{\mathbb{Q}}^p\langle e_1, \dots, e_n\rangle$ and $T_q=\mathbb{Q}[\Gamma]\otimes \bigwedge_{\mathbb{Q}}^p\langle e_{n+1}, \dots, e_r\rangle$. We can think of $K_p\otimes T_q$ as consisting of parallelepipeds with $p$ edge directions in $A$ and $q$ edge directions in $B$. 

We extend $\delta$ and $\partial$ to maps $\delta:K_{p+1}\otimes T_{q}\rightarrow K_p\otimes T_{q}$ and $\partial:K_{p}\otimes T_{q+1}\rightarrow K_p\otimes T_{q}$. In order for these to be compatible with the differential on $C_{\bullet}^n$, the map $\partial:K_{p}\otimes T_{q+1}\rightarrow K_p\otimes T_{q}$ has a $(-1)^p$ factor compared to $\partial:T_{q+1}\rightarrow T_q$. Note that $\delta\partial+\partial\delta=0$, and the differential $d$ on $C_{\bullet}^n$ is identified with $\delta+\partial$. 

Now $x$ is a cycle in $C_k^n$. For $i=0, \dots, k$, let $x_i$ be the component of $x$ in $K_i\otimes T_{k-i}$. The strategy will be to inductively add boundaries that zero $x_i$, starting with $x_0$. Note that $\varepsilon$ identifies sets of parallelepipeds with the same set of edge directions, so $\varepsilon(x)=0$ implies that $\varepsilon(x_i)=0$ for each $x_i$.

So $x_0$ is a class in $K_0\otimes T_k$ with $\varepsilon(x_0)=0$. Note that $K_0=\mathbb{Q}[\Gamma]$. The image of $\delta:K_1\rightarrow K_0$ is the ideal $I$ in $\mathbb{Q}[\Gamma]$ generated by $g_1-1, \dots, g_n-1$. We claim that $x_0\in I\otimes T_k$. Indeed, the parallelepipeds in $x_0$ can be grouped in pairs that are orientation-reversed translates of each other, and each pair is an element of $I\otimes T_k$. This says that $x_0$ is a $\delta$-boundary. Let $y_1\in K_1\otimes T_k$ with $\delta(y_1)=x_0$. We can subtract $d(y_1)=\delta(y_1)+\partial(y_1)$ from $x$ to get a new cycle with no component in $K_0\otimes T_{k}$. Note that $d(y_1)$ is a cycle (since $d^2=0$), so $x-d(y_1)$ is a cycle.

Now proceed by induction: suppose that $x_i=0$ for $i=0, \dots, j-1$ for some $j\ge 1$. We want to add a boundary to $x$ that zeros $x_j$. Now $x_j\in K_j\otimes T_{k-j}$, and $\delta(x_j)=d(x_j)$ since no other non-zero $x_i$ has any boundary in $K_{j-1}\otimes T_{k-j}$. Since $x_j$ is a cycle, $\delta(x_j)=0$. By exactness of $K_{\bullet}$, there is some $y_{j+1}\in K_{j+1}\otimes T_{k-j}$ with $\delta(y_{j+1})=x_j$. Subtracting $d(y_{j+1})=\delta(y_{j+1})+\partial(y_{j+1})$ gives a new cycle with no $K_j\otimes T_{k-j}$ component. 

By induction, we can eventually zero all the terms in $x$. Since we have only subtracted boundaries from $x$, this proves that $x$ is a boundary, as required.
\end{proof}

\begin{remark}
In \Cref{sec:auginjective}, we give a shorter, but less concrete proof of \Cref{lem:aug}. 
We are also aware of a \emph{more} explicit proof.
We have decided not to include this since the combinatorial construction is explicit.
\end{remark}
\section{ Filling Simplices}

In order to show that the $n$-simplex is in $\Lambda_n^n$, we want to find a cycle $x$ in $C_{n-2}^n$ with $\varepsilon(x)=0$ such that $\rho(x)=d(\Delta^n)$ is the boundary of the $n$-simplex. By \Cref{lem:aug}, $x$ is then the boundary of some $y\in C_{n-1}^n$, and by \Cref{lem:rhodcommutes}, $d(\rho(y))=\rho(d(y))=d(\Delta^n)$. By \Cref{lem:topinjectivity}, $\rho(y)$ is the class of $\Delta^n$ in $P_{n}^n$. Since the image of $\rho$ is in $\Lambda_{\bullet}^n$, this will show that $\Delta^n$ is in $\Lambda_{\bullet}^n$. 

For a subset $S\subseteq \{1, \dots, n+1\}$ of size $n$, let $F_S$ be the $n-1$-dimensional subspace of $\mathbb{Q}^{n}$ spanned by edge directions contained in the face of $\Delta^n$ that is the convex hull of vertices in $S$. We say that an element $x$ of $C_k^n$ satisfies the \emph{face direction condition} if, for every parallelepiped $p$ with non-zero coefficient in $x$, all the vertices of $p$ are contained in $F_S$ for some $S$.

First, we show how to realize condition $\varepsilon=0$:

\begin{lemma}\label{lem:cyclecorrection}
Let $x$ be a cycle in $C_{k}^n$ with $k\ge 2$. Then there is a cycle $x'\in C_k^n$ with $\varepsilon(x')=0$ and $\rho(x')=\rho(x)$. Furthermore, if $x$ satisfies the face direction condition, then $x'$ does as well.
\end{lemma}

\begin{proof}
Let $S:C_k^n\rightarrow C_k^n$ be the map that takes a parallelepiped in $C_k^n$, multiplies by $2$, then subdivides it into $2^k$ translated copies of the original parallelepiped (see \Cref{fig:doublingmap}). Formally,
\[S(\gamma\otimes e_{i_1}\wedge \dots\wedge e_{i_k})=\sum_{(u_1,\dots,u_k)\in \{0, 1\}^k}(\gamma^2\cdot g_{i_1}^{u_1}\cdot \ldots \cdot g_{i_k}^{u_k})\otimes e_{i_1}\wedge \dots\wedge e_{i_k}\]

Clearly, $\varepsilon(S(x))=2^k\varepsilon(x)$ for $x\in C_k^n$. We claim that $\rho(S(x))=2\rho(x)$. To see this, let $x=\gamma\otimes e_{i_1}\wedge \dots\wedge e_{i_k}$ be a parallelepiped in $C_k^n$. Recall that $\rho(x)$ is the (class of the) paramid $X$ with apex at the origin and base $\gamma+[0, t_{i_1}]+\dots+[0, t_{i_k}]$. So $2X$ is the paramid with apex at the origin and base $2\gamma+[0, 2t_{i_1}]+\dots+[0, 2t_{i_k}]$, which can be subdivided into $2^k$ pieces, corresponding to the components of $S(x)$ (see \Cref{fig:doublingmap}). 

So set $x'=\frac{1}{2^k-2}(2^kx-S(x))$. Since $k\ge 2$, the denominator is non-zero. Now $\varepsilon(x')=0$, and $\rho(x')=\rho(x)$. It is straightforward to check that $x'$ is a cycle.
\end{proof}

Now, we prove the inductive step that constructs a boundary of $\Delta^{n+1}$ given a $\Delta^{n}$. 

\begin{lemma}\label{lem:symmetrization}
Let $n\ge 2$. Suppose there is an $x\in C_{n-1}^n$ such that $\rho(x)=[\Delta^n]$, where $d(x)$ satisfies the face direction condition. Then there is a cycle $y\in C_{n-1}^{n+1}$ such that $\rho(y)=d([\Delta^{n+1}])$, where $y$ satisfies the face direction condition. 
\end{lemma}

\begin{proof}
The cycle $y$ will be constructed from one copy of $x$ for each injective map $\mu:\{1, \dots, n+1\}\rightarrow \{1, \dots, n+2\}$. We can view $\mu$ as a map that realizes $\Delta^n$ as a face of $\Delta^{n+1}$. So each facet of the $n$-simplex is represented by $(n+1)!$ symmetrized copies of $x$.

Let $\Gamma$ be the lattice generated by simplex edge directions in $\mathbb{R}^n$, and let $\Gamma'$ be the lattice generated by simplex edge directions in $\mathbb{R}^{n+1}$. A map $\mu: \{1, \dots, n+1\}\rightarrow \{1, \dots, n+2\}$ induces a homomorphism $\Gamma\rightarrow \Gamma'$. This defines a map $\psi_\mu:C_{\bullet}^n\rightarrow C_{\bullet}^{n+1}$. 







Each injective map $\mu:\{1, \dots, n+1\}\rightarrow \{1, \dots, n+2\}$ extends uniquely to a permutation $\mu':\{1, \dots, n+2\}\rightarrow \{1, \dots, n+2\}$. We define:

\[y=\frac{1}{(n+1)!}\sum_{\mu}\text{sign}(\mu')\psi_\mu(x)\]

\noindent where the sum is over injective $\mu:\{1, \dots, n+1\}\rightarrow \{1, \dots, n+2\}$.

We claim that $\rho(y)\in P_{n}^{n+1}$ is the boundary of $\Delta^{n+1}$. For each $S\subseteq \{1, \dots, n+2\}$ of size $n+1$, there are $(n+1)!$ maps $\mu:\{1, \dots, n+1\}\rightarrow \{1, \dots, n+2\}$ with image $S$. For each of these $\mu$, $\rho(\psi_{\mu}(x))$ is the $\Delta^n$ face of $\Delta^{n+1}$ in $F_S$, with sign determined by $\text{sign}(\mu')$ (note that the orientation on $F_S$ induced by $\mu$ is determined by $\text{sign}(\mu')$). So $\rho(y)$ is exactly the boundary of $\Delta^{n+1}$.

Next, we claim that $y$ is a cycle in $C_{n-1}^{n+1}$. Clearly, $\psi_{\mu}$ commutes with the boundary $d$ on $C_{n-1}^{n+1}$. For each parallelepiped $z$ with non-zero coefficient in $d(x)$, we choose some $S\subseteq \{1, \dots, n+1\}$ such that $z$ is in $F_S$. This is possible by the face direction condition for $d(x)$. 

Let $z$ be a parallelepiped in $d(x)$, which by assumption is in $F_S$ for some $S$. Let $\mu:\{1, \dots, n+1\}\rightarrow \{1, \dots, n+2\}$ be injective, and let $\kappa$ be the unique map that agrees with $\mu$ on $S$ but is not equal to $\mu$. Note that $\text{sign}(\mu')=-\text{sign}(\kappa')$, since $\mu'$ and $\kappa'$ differ by a transposition. Now $\psi_\mu(z)=\psi_{\kappa}(z)$, so these terms cancel in $d(y)$. This shows that $d(y)=0$. 

Note that every parallelepiped in the image of $\psi_{\mu}$ is contained in the subspace $F_{\text{img}(\mu)}$, so $y$ satisfies the face direction condition.
\end{proof}

\begin{theorem}\label{thm:homologymain}
For each $n$, $[\Delta^n]\in \Lambda^n_n$.
\end{theorem}

\begin{proof}
We will show by induction on $n$ that there is an $x\in C_{n-1}^n$ such that $\rho(x)=\Delta^n$, and $d(x)$ satisfies the face direction condition. The claim will follow since $\text{img}(\rho)\subseteq \Lambda_{\bullet}^n$.

The base case is $n=3$. We notice that the $3$ simplex $\Delta^3$ is equivalent in $\widetilde{V}_{\mathbb{R}^n}$ to the negation of a paramid $X$ that satisfies the face direction conditions. This is illustrated in \Cref{fig:3simplexrep}. The two polytopes give a subdivision of a polytope which is the sum of a triangle and a line segment. Since the triangle and line segment are less than $3$ dimensional, this shows that $\Delta^3+X=0$ in $\widetilde{V}_{\mathbb{R}^n}$. It is straightforward to check the face direction conditions for $X$. So the class of the base of $X$ (with appropriately chosen orientation) gives an $x\in C_{2}^3$ such that $\rho(x)=\Delta^3$.

\begin{figure}
\centering
\includegraphics[page=4]{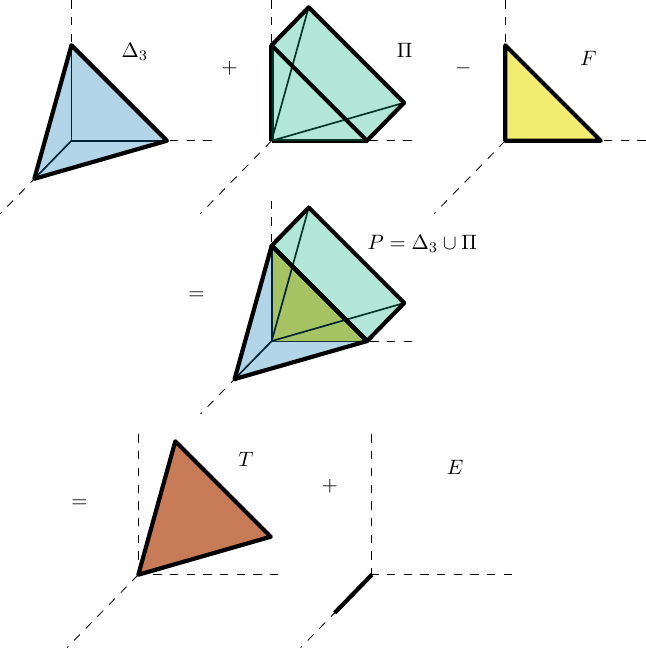}
\caption{ Top: Using inclusion-exclusion, we can write the union $P$ of the standard $3$-simplex $\Delta_3$ -- with vertices at $(0, 0, 0), (1, 0, 0), (0, 1, 0)$ and $(0, 0, 1)$ -- and the paramid $\Pi$ -- with vertices $(0, 0, 0), (0, 1, 0), (-1, 1, 0), (-1, 0, 1)$  and $(0, 0, 1)$ -- as the sum of the simplex and the paramid minus their intersection -- the triangular face $F$ with vertices $(0,0,0)$, $(0,1,0)$ and $(0,0,1)$. Bottom: $P$ is also the Minkowski sum of the triangular face $T$ with vertices at $(0, 0, 0), (-1, 1, 0)$ and $(-1, 0, 1)$ and the segment $E$ with vertices at $(0, 0, 0)$ and $(1, 0, 0)$.}
\label{fig:3simplexrep}
\end{figure}

For the inductive step, suppose the claim is true for some $n\ge 3$. By \Cref{lem:symmetrization}, there is a cycle $y\in C_{n-1}^{n+1}$ such that $\rho(y)=d(\Delta^{n+1})$. By \Cref{lem:cyclecorrection}, there is some $y'$ with $\rho(y')=\rho(y)=d(\Delta^{n+1})$, with $\varepsilon(y')=0$. So by \Cref{lem:aug}, there is some $z\in C_{n}^{n+1}$ with $d(z)=y'$ (so $d(z)$ satisfies the face direction condition).

By \Cref{lem:rhodcommutes}, we have $d(\rho(z))=\rho(d(z))=\rho(y')=d(\Delta^{n+1})$. By \Cref{lem:topinjectivity}, this means that $\rho(z)=\Delta^{n+1}$. By induction, this completes the proof.
\end{proof}

\section{ Additional Ideas }

This section contains some results and alternate proofs that did not fit into the main body of this part of the paper. We believe that these ideas can be valuable intuition for understanding our result.

\subsection{ The Universality of Paramids }\label{sec:paramidscomplete}

Haase, Hertrich, and Loho show that $\relu_{n,2}$ is generated by support functions of polytopes that are each the convex hull of two zonotopes (see Theorem 8 in \cite{haaseLowerBoundsDepth2023}). By Propositions 9 and 10 in \cite{haaseLowerBoundsDepth2023} and \Cref{lem:additivity}, such a polytope can be written in terms of polytopes of form $\text{conv}(A\cup B)$, where $\text{conv}(A \cup  B)$ is a \emph{join} of $A$ and $B$ (i.e. $A$ and $B$ have disjoint affine hulls and their tangent spaces are transverse). Working modulo lower-dimensional pieces, we can ignore pieces except those with $\dim(A)+\dim(B)=n-1$.

When $A$ or $B$ is a point, the piece we get is a paramid. The purpose of this section is to show that, modulo lower-dimensional terms, \emph{every} such piece is equivalent to a paramid. This somewhat justifies the choice to focus on paramids with our proof.

We write $\alpha:\mathbb{R}^n\rightarrow \mathbb{R}^n$ for the antipodal map. The key result is the following:

\begin{lemma}\label{lem:joinsumlowerdim}
Let $A$ and $B$ be polytopes in $\mathbb{R}^n$ ($n\ge 1$) with disjoint affine hulls. Suppose $A$ can be written $A_1+A_2$, where $A_1$ and $A_2$ have transverse tangent spaces. Then in $\widetilde{V}_{\mathbb{R}^n}$, the class of $\conv(A\cup B)$ is equal to $(-1)^{\dim(A_2)}$ times the class of $\conv(A_1\cup (B+\alpha(A_2)))$.
\end{lemma}

\begin{proof}
In this proof, we use the convention that a zero-dimensional point has one $-1$-dimensional facet, which is the empty set. The Minkowski sum of any polytope with the empty set is empty. The way we use the induction hypothesis, we will sometimes need a version of the claim where $A_1$ or $B$ empty. We will check this case first. 

If both $A_1$ and $B$ are empty, then both $\text{conv}(A\cup B)$ and $\text{conv}(A_1\cup (B+\alpha(A_2)))$ are empty, so the claim is trivial. If $A_1$ is empty and $B$ is not, then $\text{conv}(A\cup B)=B$ and $\text{conv}(A_1\cup (B+A_2))=B+A_2$. Since $A_2$ and $B$ have dimension $<n$, these both represent $0$ in $\widetilde{V}_{\mathbb{R}^n}$.

We also handle the case $\dim(A_2)=0$ separately. In this case, $A_2$ is a point, and it is clear that $\text{conv}(A_1\cup (B+\alpha(A_2)))$ is just $\text{conv}(A\cup B)$ translated by $-A_2$.

Now suppose that $A$ and $B$ are non-empty and $A_2$ has dimension at least $1$. The proof is by induction on $n$. When $n=1$, the conditions of the lemma are impossible to satisfy, and so the result is true trivially. 

We want to use \Cref{lem:lowerdimcharacterization}. Since $A$ and $B$ have disjoint affine hulls, the facets of $\text{conv}(A\cup B)$ are of form $\text{conv}(F\cup B)$ or $\text{conv}(A\cup G)$, where $F$ is a facet of $A$ or $G$ is a facet of $B$. 

Since $A_1$ and $A_2$ have transverse tangent spaces, the facets of $A=A_1+A_2$ are of form $F_1+A_2$ or $A_1+F_2$, where $F_1$ is a facet of $A_1$ or $F_2$ is a facet of $A_2$. So the facets of $\text{conv}(A\cup B)$ are:

\begin{itemize}
    \item for each facet $G$ of $B$, $\text{conv}(A\cup G)$,
    \item for each facet $F_1$ of $A_1$, $\text{conv}((F_1+A_2)\cup B)$, and
    \item for each facet $F_2$ or $A_2$, $\text{conv}((A_1+F_2)\cup B)$.
\end{itemize}

Since $A$ and $B$ have transverse tangent spaces, $B$ and $A_2$ do as well. So by a similar argument to the above, the facets of $\text{conv}(A_1\cup (B+\alpha(A_2)))$ are:

\begin{itemize}
    \item for each facet $F_1$ of $A_1$, $\text{conv}(F_1\cup (B+\alpha(A_2)))$,
    \item for each facet $G$ of $B$, $\text{conv}(A_1\cup (G+\alpha(A_2)))$
    \item for each facet $F_2$ of $A_2$, $\text{conv}(A_1\cup (B+\alpha(F_2)))$
\end{itemize}

By the inductive hypothesis, the class (in $\widetilde{V}_S$, where $S$ contains $A_1, G$ and $A_2$) of $\text{conv}(A_1\cup (G+\alpha(A_2)))$ is equal to $(-1)^{\dim(A_2)}$ times the class $\text{conv}(A\cup G)$. Similarly, the class of $\text{conv}(F_1\cup (B+\alpha(A_2)))$ is equivalent to $(-1)^{\dim(A_2)}$ times the class of $\text{conv}((F_1+A_2)\cup B)$. The dimension of $F_2$ is one less than the dimension of $A_2$, so the class of $\text{conv}(A_1\cup (B+\alpha(F_2)))$ is $-(-1)^{\dim(A_2)}$ times the class $\text{conv}((A_1+F_2)\cup B)$.

We notice that the outward facing normal vector of $\text{conv}(A_1\cup (G+\alpha(A_2)))$ is the same as the one for $\text{conv}(A\cup G)$, and similarly for $\text{conv}(F_1\cup (B+\alpha(A_2)))$ and $\text{conv}((F_1+A_2)\cup B)$, while $\text{conv}(A_1\cup (B+\alpha(F_2)))$ and $\text{conv}((A_1+F_2)\cup B)$ have opposite normal vectors. By \Cref{lem:lowerdimcharacterization}, this completes the proof. 
\end{proof}

In particular, if $A$ and $B$ are zonotopes, then the class of $\text{conv}(A\cup B)$ is, up to lower dimensional terms, equal to $(-1)^{\dim(A)}$ times the class of $\text{conv}(\{0\}, B+\alpha(A))$. The latter polytope is a paramid. 

So, up to lower dimensional terms, support functions of paramids generate $\relu_{n,2}$. Of course, this follows from \Cref{thm:homologymain} as well.

\begin{remark}
Even if polytopes $A$ and $B$ in $\mathbb{R}^n$ are equivalent modulo terms of dimension $<n$, it is not true in general that pyramids on $A$ and $B$ in $\mathbb{R}^{n+1}$ are equivalent modulo terms of dimension $<n+1$. So \Cref{lem:joinsumlowerdim} is not so useful for inductive arguments.
\end{remark}

\subsection{ The Koszul Complex }\label{sec:koszul}

The chain complex $C_{\bullet}^{n}$ that we define is also known as the \emph{Koszul complex} of the sequence $t_1-1, \dots, t_r-1$, viewing $t_1-1, \dots, t_r-1$ as elements of $\mathbb{Q}[\Gamma]$. 

Formally, given a ring $U$ and a set of elements $x_1, \dots, x_r$, the associated Koszul complex $K_{\bullet}$ is defined by $K_{p}=\bigwedge_{U}^p U^r$ (so $K_0=U$). Viewing $U^r$ as being generated by $e_1, \dots, e_r$, the differential is defined by $d(e_i)=x_i$ and graded Leibniz. Note that, if $U$ is an $\mathbb{Q}$-module (e.g. $\mathbb{Q}[\Gamma]$), then $K_p$ is naturally isomorphic to $U\otimes_{\mathbb{Q}}\bigwedge_{\mathbb{Q}}^p\mathbb{Q}^r$.
 
A sequence $F=(f_1, \dots, f_r)$ is called a \emph{regular sequence} if the ideal $(f_1, \dots, f_r)$ is not equal to $U$ and, for each $i$, $f_i$ is not a zero-divisor in $R/(f_1, \dots, f_{i-1})$. If $f_i=t_i-1$ as in \Cref{sec:parallelepipeds}, then regularity says that the translations $t_i$ freely generate an abelian subgroup of $\mathbb{Q}^n$. 

The homology of a Koszul complex of a regular sequence is well understood:

\begin{theorem}(See e.g. \cite{Weibel1994}, Corollary 4.5.5)
Let $U$ be a ring, let $f_1, \dots, f_r$ be a regular sequence, and let $K_{\bullet}$ be the Koszul complex. Then $H_0(K_{\bullet})=U/(f_1, \dots, f_r)$, and all higher homology is zero. 
\end{theorem}

\subsection{ The Action of $S$ on $H_\bullet(C_{\bullet}^n)$ }\label{sec:tori}

Let $S$ be the doubling map defined in the proof of \Cref{lem:cyclecorrection} (see also \Cref{fig:doublingmap}). Instead of using the augmentation condition for cycle filing in $C_{\bullet}^n$, we could instead try to compute the action of $S$ on $H_{\bullet}(C_{\bullet})$. Note that $S$ computes with the boundary on $C_{\bullet}^n$, so it induces a map $S_*:H_{\bullet}(C_{\bullet}^n)\rightarrow H_{\bullet}(C_{\bullet}^n)$.

Recall from \Cref{sec:homology} that $C_{\bullet}^n$ is the quotient of the unit-cubical complex $\hat{C}_{\bullet}^n$ of $\mathbb{R}^r$ by the subgroup of $\mathbb{Z}^r$ representing relations among the simplex edge directions $t_1, \dots, t_r$. Now $S$ lifts to a map $\hat{C}_{\bullet}^n\rightarrow \hat{C}_{\bullet}^n$ in a straightforward way.

Note that $\hat{C}_{\bullet}^n$ is a cubical complex for $\mathbb{R}^r$ and $C_{\bullet}^n$ is a cubical complex for the product $\mathbb{R}^{n}\times (S^1)^{r-n}$. The map $\pi:\mathbb{R}^r\rightarrow \mathbb{R}^{n}\times (S^1)^{r-n}$ is the universal covering map of $\mathbb{R}^{n}\times (S^1)^{r-n}$.

Now $S$ gives a map $\mathbb{R}^{n}\times (S^1)^{r-n}\rightarrow \mathbb{R}^{n}\times (S^1)^{r-n}$, which lifts to a map $S:\mathbb{R}^r\rightarrow \mathbb{R}^r$, which is given by multiplication by $2$. 

Now $\pi_1(\mathbb{R}^{n}\times (S^1)^{r-n})$ is naturally associated with the group of deck transformations of $\mathbb{R}^r$, which is the kernel of $\pi: \hat{\Gamma}\rightarrow \Gamma$. So we see that $S_{*}:\pi_1(\mathbb{R}^{n}\times (S^1)^{r-n})\rightarrow \pi_1(R^n\times (S^1)^{r-n})$ is multiplication by $2$. 

So $S_{*}:H_1(C_{\bullet}^n)\rightarrow H_1(C_{\bullet}^n)$ is multiplication by $2$. By functoriality of the product structure on the (co)homology of $(S^1)^{r-n}$, we see that $S_{*}:H_k(C_{\bullet})\rightarrow H_k(C_{\bullet})$ is multiplication by $2^k$. This gives us a different way of seeing that, for a cycle $x\in C_{k}^n$, $2^kx-S(x)$ is a boundary.

\subsection{ An Alternative Proof of \Cref{lem:aug} }\label{sec:auginjective}

We can make $D_{\bullet}^n$ into a chain complex with boundary map $0$. The map $\varepsilon: C_{\bullet}^n\rightarrow D_{\bullet}^n$ commutes with the boundary maps, so descends to a map $\varepsilon_*$ on homology. Since the boundary on $D_{\bullet}^n$ is $0$, the homology of $D_{\bullet}^n$ is just $D_{\bullet}^n$. 

Recall that \Cref{lem:aug} says that $\varepsilon: C_{\bullet}^n\rightarrow D_{\bullet}^n$ sends a cycle $x$ to $0$ if and only $x$ is a boundary. This is equivalent to saying that $\varepsilon_*:H_{\bullet}(C_{\bullet}^n)\rightarrow H_{\bullet}(D_{\bullet}^n)$ is injective.

Now $D_{\bullet}^n$ is a cubicalization of the quotient of $\mathbb{R}^r$ by the action of $\mathbb{Z}^r$. Recall that $C_{\bullet}^n$ is the cubicalization of the quotient of $\mathbb{R}^r$ by the group of relations among the $t_i$, call this group $G$. Now $\varepsilon$ corresponds to the natural covering map $\mathbb{R}^r/G\rightarrow \mathbb{R}^r/\mathbb{Z}^n$. 

Now $H_k(C_{\bullet}^n)$ is identified with $\bigwedge_{\mathbb{Q}}^k(\mathbb{Q}\otimes G)$ and $H_{k}(D_{\bullet}^n)$ is identified with $\bigwedge_{\mathbb{Q}}^k(\mathbb{Q}\otimes \mathbb{Z}^r)$, and $\varepsilon_*$ is the map $\bigwedge_{\mathbb{Q}}^k(\mathbb{Q}\otimes G)\rightarrow \bigwedge_{\mathbb{Q}}^k(\mathbb{Q}\otimes \mathbb{Z}^r)$ induced by the inclusion $G\rightarrow \mathbb{Z}^r$. This map is injective, giving an alternate proof of \Cref{lem:aug}.

\subsection{ An Alternate Proof of \Cref{lem:lowerdimcharacterization} }\label{sec:flaglowerdim}

A result of Koutschan, Ponomarchuk, and Schicho \cite{KPS2026} gives an exact condition for when a continuous piecewise linear function $f$ in $n$ inputs can be written as a linear combination of ``lower-dimensional'' functions. Here we show how to derive \Cref{lem:lowerdimcharacterization} from their result. The results of \cite{KPS2026} are stated over $\mathbb{R}$, so in this section we work with $\mathbb{R}$ coefficients instead of $\mathbb{Q}$.

The results of \cite{KPS2026} are stated over the reals. We expect that the results also hold with rational coefficients, but in this section we will allow arbitrary real coefficients. 

Let $S$ and $Q$ be finite dimensional real vector spaces, and let $g:S\rightarrow Q$ be a piecewise constant function. We allow $g$ to be defined only ``almost everywhere'' (precisely, we say that for ``almost all'' $x\in S$, $g(x)$ is defined and constant on a neighborhood of $x$). Let $H\subseteq S$ be an oriented codimension $1$ subspace. For $x\in H$ and $\epsilon>0$, let $B_\epsilon^+(x)$ and $B_{\epsilon}^-(x)$ be the two halves of the ball of radius $\epsilon$ around $x$, split by $H$, with $B_\epsilon^+(x)$ being the half on the positive side of $H$. For almost all $x\in H$, $g$ is constant on $B_\epsilon^+(x)$ and $B_\epsilon^-(x)$. We define $\Delta_g(H):H\rightarrow Q$ to be the piecewise constant, almost everywhere defined function that is the difference of the values of $g$ on $B_\epsilon^+$ and $B_\epsilon^-$, when this exists. 

An oriented flag $\mathcal{H}=(H_1, \dots, H_k)$ of length $k$ in $S$ is a sequence of oriented codimension $1$ affine subspaces $H_1\subseteq S$ and $H_{i+1}\subseteq H_i$. For $g:S\rightarrow Q$ piecewise constant and almost everywhere defined, we define $\Delta_g(\mathcal{H})$ inductively as follows: we define $g_1=\Delta_{g}(H_1)$ and $g_{i+1}=\Delta_{g_i}(H_{i+1})$, and set $\Delta_g(\mathcal{H})=g_k$. We say that $\mathcal{H}$ is \emph{critical} for $g$ if $g_k$ is non-zero.

Let $\mathcal{H}=(H_1, \dots, H_k)$ be a flag, and let $v_1, \dots, v_k$ be vectors such that $v_i$ is the normal vector to $H_i$ in $H_{i-1}$, viewing $H_0$ as the base space $S$. Then, where $\Delta_g(\mathcal{H})$ is defined, it is equal to:

\[\Delta_g(\mathcal{H})(x)=\lim_{\epsilon_k\rightarrow 0^+}\dots\lim_{\epsilon_1\rightarrow 0^+}\sum_{(s_1, \dots, s_k)\in \{-1, 1\}^{k}}g(x+s_1\epsilon_1v_1+\dots+s_k\epsilon_kv_k)\cdot\prod_{i=1}^ks_k\]

If $f:\mathbb{R}^n\rightarrow \mathbb{R}$ is a continuous piecewise linear function, then the gradient $\nabla f:\mathbb{R}^n\rightarrow \mathbb{R}^n$ is a piecewise constant, almost everywhere defined function. The function $f$ is said to be a $\max_k$ if it can be written $f=\max(\ell_1, \dots, \ell_k)$, where each $\ell_i$ is an affine function. The following is proved in \cite{KPS2026}:

\begin{theorem}(\cite{KPS2026}, Theorem 1)\label{thm:flagcondition}
Let $f:\mathbb{R}^n\rightarrow \mathbb{R}$ be a continuous piecewise linear function. Then we can express $f$ as a linear combination of $\max_{k+1}$ functions if and only if every critical flag for $\nabla f$ has length $\le k$.
\end{theorem}

A piecewise linear function is said to be homogeneous if every piece extends to a linear function (i.e. an affine function that is $0$ at $0$). Given a polytope $X$ in $\mathbb{R}^n$, its support function $h_X:\mathbb{R}^n\rightarrow \mathbb{R}$ is defined by $h_X(v)=\max_{x\in X}(v\cdot x)$. This is a homogeneous piecewise linear function. There is a linear map from $V_S$ to the space of homogeneous piecewise linear functions on $S$ that sends $(X, Y)$ to $(h_X, h_Y)$. This is an isomorphism.

A \emph{homogeneous} $\max_{k+1}$ (that is, a max of $k$ linear functions) is the support function of a polytope with $k+1$ vertices, which is at most $k$ dimensional. So \Cref{thm:flagcondition} is close to giving a characterization of $\widetilde{V}_S$. We need one more lemma:

\begin{lemma}\label{lem:homogeneouscorrection}
If $f:\mathbb{R}^n\rightarrow \mathbb{R}$ is a homogeneous piecewise linear function that can be written as a sum $\max_k$ functions, then $f$ can be written as a sum of homogeneous $\max_k$ functions.
\end{lemma}

\begin{proof}
We write:

\[f+\sum_{i=1}^ka_i-\sum_{i=1}^jb_i=0\]

Where each $a_i$ and $b_i$ is a $\max_k$ term. Now obtain $a_i'$ from $a_i$ by adding a constant so that $a_i'(0)=0$. Since $f(0)=0$, all the constant terms added must cancel, so:

\[h_X-h_Y+\sum_{i=1}^ja_i'-\sum_{i=1}^m b_i'=0\]

Now each $a_i'$ or $b_i'$ is a max of $n$ terms, plus a constant. Write $a_1'=c_0+\max(\ell_1+c_1, \dots,\ell_n+c_n)=c_0+\max(\ell_1+c_0+c_1, \dots,\ell_n+c_0+c_n)$, where each $\ell_i$ is linear and each $c_i$ is a constant. We know that $a_1(0)=0$, so all the $c_0+c_i$ are $\le 0$. Obtain $a_1''$ from $a_1$ by removing all the $\ell_i+c_0+c_i$ terms that have $c_0+c_i<0$. Note that $a_1''=a_1'$ in a neighborhood of $0$. We similarly define $a_i''$ and $b_i''$ for each $a_i$, and notice that:

\[f+\sum_{i=1}^ja_i''-\sum_{i=1}^mb_i''=0\]

\noindent in a neighborhood of $0$. However, both sides are homogeneous, and a homogeneous piecewise linear function is determined by its behavior near $0$, so both sides are equal. 
\end{proof}

We now have a characterization of $\widetilde{V}_S$. Let $(X, Y)$ represent a class in $\widetilde{V}_S$ and let $g=\nabla (h_X-h_Y)$. Then $(X, Y)=0$ in $\widetilde{V}_S$ if, for every flag $\mathcal{H}$ of length $n$, $\Delta_g(\mathcal{H})=0$. The converse is also true, but this won't be relevant for our approach.

We say that a flag $\mathcal{H}(H_1, \dots, H_k)$ is homogeneous if every $H_i$ is a vector space (i.e. it contains $0$). The following shows that, if we want to apply \Cref{thm:flagcondition} to a homogeneous function, then we only need to consider homogeneous flags.

\begin{lemma}\label{lem:homogeneousflag}
If $f$ is a homogeneous piecewise linear function, $g=\nabla f$ and $\mathcal{H}$ is a flag, then $\Delta_g(\mathcal{H})=0$ unless $\mathcal{H}$ is homogeneous.
\end{lemma}

\begin{proof}
Straightforward.
\end{proof}

\begin{lemma}\label{lem:flaginduction}
Let $X$ be a polytope in $\mathbb{R}^n$. Let $\mathcal{H}=\{H_1, \dots, H_{n}\}$ be a homogeneous flag (so $H_n=0$). Let $v_1, \dots, v_n$ be vectors such that $v_i$ is a normal vector to $H_i$ in $H_{i-1}$, viewing $H_0=\mathbb{R}^n$. Let $F_1$ and $F_2$ be the faces of $X$ in directions $v_n$ and $-v_n$ respectively. Let $S$ be the orthogonal complement of $v_n$ in $\mathbb{R}^n$. For $i=1, \dots, n-1$, let $H_i'$ be the orthogonal complement of $v_n$ in $H_i$, so $\mathcal{H}'=\{H_1', \dots, H_{n-1}'\}$ is a flag on $S$. Let $f=\nabla h_X$ and let $f'=\nabla h_{F_1}-\nabla h_{F_2}$, where we restrict $f'$ to a function $S\rightarrow \mathbb{R}^n$. Then $\Delta_f(\mathcal{H})(0)=\Delta_{f'}(\mathcal{H}')(0)$.
\end{lemma}

\begin{proof}
We know that:

\[\Delta_f(\mathcal{H})(0)=\lim_{\epsilon_n\rightarrow 0^+}\dots\lim_{\epsilon_1\rightarrow 0}\sum_{(s_1, \dots, s_n)\in \{-1, 1\}^{n}}g(s_1\epsilon_1v_1+\dots+s_n\epsilon_nv_n)\cdot \prod_{i=1}^ns_i\]

\noindent and:

\[\Delta_{f'}(\mathcal{H}')(0)=\lim_{\epsilon_{n-1}\rightarrow 0^+}\dots\lim_{\epsilon_1\rightarrow 0^+}\sum_{(s_1, \dots, s_{n-1})\in \{-1, 1\}^{n-1}}f'(s_1\epsilon_1v_1+\dots+s_{n-1}\epsilon_{n-1}v_{n-1})\cdot \prod_{i=1}^{n-1}s_i\]

Note that $h_X$, $h_{F_1}$ and $h_{F_2}$ are homogeneous, so $f(\lambda x)=f(x)$ for $\lambda>0$, and similarly for $f'$. We observe that, for almost all $x\in S$:

\[\nabla h_{F_1}(x)=\lim_{\epsilon\rightarrow 0^+}f(v_n+\epsilon x)\]

To see this, we notice that, for $v$ a direction, $\nabla f(v)$ is equal to the vertex $u$ of $X$ that maximizes $u\cdot v$ ($\nabla f(x)$ is defined when this is unique). Now $F_1$ is the set of points $u$ in $X$ maximizing $u\cdot v_n$. So when $x$ is sufficiently small, $f(v_n+x)$ is the vertex $u$ of $F_1$ maximizing $u\cdot x$, which is equal to $h_{F_1}(x)$. 

Similarly, for almost all $x\in S$, we have:

\[\nabla h_{F_2}(x)=\lim_{\epsilon\rightarrow 0^+}f(-v_n+\epsilon x)\]

So for almost all $x\in S$, we have:

\[f'(x)=\lim_{\epsilon\rightarrow 0^+}(f(v_n+\epsilon x)-f(-v_n+\epsilon_x))\]

It follows that $\Delta_f(\mathcal{H})(0)=\Delta_{f'}(\mathcal{H}')(0)$.
\end{proof}

\begin{proof}[Proof of the ``if'' direction of \Cref{lem:lowerdimcharacterization}]
Let $X$ and $Y$ be polytopes in $\mathbb{R}^n$. Suppose that, for every direction $v$, $(F_1+G_2, F_2+G_1)=0$ in $\widetilde{V}_S$, where $S$ is the orthogonal complement of $v$,  $F_1$ and $F_2$ are the faces of $X$ in directions $v$ and $-v$ respectively, and $G_1$ and $G_2$ are faces of $Y$ in directions $v$ and $-v$ respectively.

We want to show that there are $A_i$ and $B_i$ such that $X+A_1+\dots+A_k=Y+B_1+\dots+B_j$. Let $f=\nabla (h_X-h_Y)$. By \Cref{thm:flagcondition} and \Cref{lem:homogeneouscorrection,lem:homogeneousflag}, it is sufficient to show that, for any homogeneous flag $\mathcal{H}=\{H_1, \dots, H_n\}$, we have $\Delta_{f}(\mathcal{H})(0)=0$. This is by induction on $n$.

Let $\mathcal{H}=\{H_1, \dots, H_n\}$ be a flag and let $v$ be a vector in $H_n$ representing the orientation of $H_{n-1}$. Let $F_1, F_2, G_1$ and $G_2$ be faces of $X$ and $Y$ in directions $\pm v$, as in the statement of the lemma, but \emph{not} translated. Let $f'=\nabla (h_{F_1}+h_{G_2}-h_{F_2}-h_{G_1})$. By \Cref{lem:flaginduction}, we have $\Delta_{f'}(\mathcal{H}')(0)=\Delta_f(\mathcal{H})(0)$, where $\mathcal{H}'$ is the flag obtained as in \Cref{lem:flaginduction}.

Let $F_1', F_2', G_1'$ and $G_2'$ be $F_1$ etc. translated to lie in the subspace $S$ orthogonal to $v$. Let $f''=\nabla (h_{F_1'}+h_{G_2'}-h_{F_2'}-h_{G_1'})$. Translating a polytope changes its support function by adding a linear function, so $f''$ differs from $f'$ by a constant. Since $\mathcal{H}'$ has length $\ge n-1\ge 1$, we have $\Delta_{f''}(\mathcal{H}')(0)=\Delta_{f'}(\mathcal{H}')(0)$. 

We know that $(F_1'+G_2', F_2'+G_1')=0$ in $\widetilde{V}_S$ by assumption. In the base case $n=2$, this means that $F_1'+G_2', F_2'+G_1'$ are translates of each other, so $f''$ is constant and $\Delta_{f''}(\mathcal{H}')=0$. When $n>2$, the ``only if'' part of \Cref{lem:lowerdimcharacterization} shows that $(F_1'+G_2', F_2'+G_1')$ satisfies the conditions of the lemma. So we know by inductive assumption that $\Delta_{f''}(\mathcal{H}')(0)=0$. In either case, we conclude that $\Delta_f(\mathcal{H})(0)$, so $h_X-h_Y$ can be written in terms of homogeneous $\max_n$ terms, which means that $(X, Y)=0$ in $\widetilde{V}_{\mathbb{R}^n}$.
\end{proof}

\part{A Direct Combinatorial Proof}
Here we explicitly construct a signed sum of paramids whose support functions equal \(\maxfn{n}\) modulo lower-dimensional terms. 

\section{Preliminaries}
Throughout this part, let $\upto{n} \coloneqq \set{1, \ldots, n}$ and let
\[
  \Delta
  \coloneqq
  \conv\set{e_1, \ldots, e_n}
  \subset
  H
  \coloneqq
  \rset*{x \in \R^n}{\sum_i x_i = 1}
\]
be the standard $(n-1)$-simplex living in the affine subspace $H$ such that $h_\Delta = \maxfn{n}$. For every index $v \in \upto{n}$, let
\[
  V_{-v} \coloneqq \upto{n} \setminus \set{v}, \qquad
  H_{-v} \coloneqq \rset{p \in H}{p_v = 0}, \qquad
  \Delta_{-v} \coloneqq \conv\rset{e_i}{i \in V_{-v}}.
\]
Furthermore, let $\cL_{n-2}$ be the span of support functions of polytopes in $H$ of dimension at most $n-2$. Write $f \equiv g$ if $f-g \in \cL_{n-2}$.
We write \enquote{almost everywhere} to mean \enquote{outside a finite union of lower dimensional spaces}.
For a polytope $P$ and a point $q$, the pyramid with apex $q$ and base $P$ is denoted as
\[
  \pyr{q}{P} \coloneqq \conv(\set{q} \cup P).
\]

\subsection{Standard Valuations}
\cref{lem:additivity} gives an inclusion--exclusion identity for support functions of a polytope covered by polytopes. The following lemma uses the same valuation property to handle arbitrary linear relations between indicator functions that hold almost everywhere, yielding an identity of support functions modulo lower-dimensional terms.
\begin{lemma}[Valuation]
  \label{lem:valuation}
  Let \(A \subseteq H\) be a \(d\)-dimensional affine subspace, let \(P_1,\ldots,P_k\subseteq A\) be polytopes, and let \(c_1,\ldots,c_k\in\R\).
  If \(\sum_{i=1}^k c_i\ind{P_i}=0\) almost everywhere on \(A\), then
  \[
    \sum_{i=1}^k c_i h_{P_i}
  \]
  is a linear combination of support functions of polytopes in \(A\) of dimension at most \(d-1\).
  Moreover, for every \(q\in H\setminus A\),
  \[
    \sum_{i=1}^k c_i h_{\pyr{q}{P_i}}
  \]
  is a linear combination of support functions of polytopes in \(\aff(A\cup\{q\})\) of dimension at most \(d\).
\end{lemma}

\begin{proof}
  The first assertion follows from full additivity of support functions after a suitable subdivision (see \Cref{lem:additivity} and e.g.~\cite{Schneider2013}).
  For the second, fix \(q\in H\setminus A\) and set \(B=\operatorname{aff}(A\cup\{q\})\).
  Every section of the pyramids parallel to \(A\), apart from the apex, is a common homothetic image of the original polytopes. Their signed indicators therefore cancel almost everywhere on \(B\).
  Apply the first assertion in the \((d+1)\)-dimensional affine space \(B\).
\end{proof}

\section{Replacing Simplices by Paramids}
\subsection{The Parallelepipeds}
Here we describe the parallelepipeds that form the bases of the paramids in our construction.
For $v \in \upto{n}$, an ordering $\pi = (\pi_{1}, \ldots, \pi_{n-1})$ of $V_{-v}$, and a set $M \subseteq \set{2,\ldots,n-1}$ of \emph{marked} indices, define the polytope
\begin{equation}\label{eq:definition-parallelpipeds}
  B_{v,\pi,M}
  \coloneqq
  \rset*{p \in H_{-v}}{
    \begin{array}{@{}r@{\;}c@{\;}c@{\;}c@{\;}l@{\quad}l@{}}
      0  & \le & u_i       & \le & 1 & \forall i\in M,                        \\
      -1 & \le & p_{\pi_i} & \le & 0 & \forall i\in\set{2,\dots,n-1}\setminus M
    \end{array}
  },
\end{equation}
where $u_i = \sum_{j<i}p_{\pi_j}$.

\begin{claim}
  The set $B_{v,\pi,M}$ is a parallelepiped.
\end{claim}

\begin{proof}
  Since $p_v=0$ and $\sum_i p_i=1$, the coordinates $p_{\pi_2},\ldots,p_{\pi_{n-1}}$ determine $p$ uniquely.
  The change of coordinates
  \[
    t_i=\begin{cases}
      u_i=1-\sum_{j=i}^{n-1} p_{\pi_j} & i\in M,     \\
      -p_{\pi_i}                   & i\notin M
    \end{cases}
  \]
  is triangular with diagonal entries $-1$, and hence invertible.
  In these coordinates, $B_{v,\pi,M}$ is the cube $[0,1]^{n-2}$.
\end{proof}

\subsection{The Replacement Lemma}
This is the main step that allows us to replace simplices by paramids.
For $v \in \upto{n}$ and every nonempty $R \subseteq V_{-v}$, define the dilated simplex
\[
  Q_{v,R}
  =(n-\abs{R})\Delta_{-v}-\sum_{i\notin R \cup \set{v}}e_i \subset H_{-v},
\]
a dilated and translated copy of the facet $\Delta_{-v}$. Recall \Cref{fig:simpAndPar}.

\begin{lemma}\label{lem:replacement}
  For each $v\in\upto{n}$, the following identity holds:
  \begin{equation*}
    \sum_{\pi}\sum_{M \subseteq \set{2, \ldots, n-1}}
    (-1)^{\abs{M}}
    h_{\pyr{e_v}{B_{v,\pi,M}}}
    \equiv
    (n-1)!\sum_{\emptyset\ne R\subseteq V_{-v}}
    \frac{(-1)^{\abs{R}-1}\abs{R}}{n-\abs{R}}
    h_{\pyr{e_v}{Q_{v,R}}},
  \end{equation*}
  where $\pi$ ranges over all orderings of $V_{-v}$.
\end{lemma}

\subsection{The Thermostat Cancellation}
Fix $v \in \upto{n}$. We prove that the signed sum of indicator functions over the parallelepipeds $B_{v, \pi, M}$ defined in \cref{eq:definition-parallelpipeds} has a remarkably simple expression, when summing over all permutations $\pi$ and sets of markings $M$, that only depends on the number of positive entries of the argument.
\begin{lemma}\label{lem:thermostatN}
  For every $v \in \upto{n}$ and almost every $p$ in $H_{-v}$,
  \[
    \sum_{\pi}\sum_{M\subseteq\set{2, \ldots, n-1}}
    (-1)^{\abs{M}}\ind{B_{v,\pi,M}}(p)
    =
    \begin{cases}
      (-1)^{s-1}s!(n-1-s)!
        & p_i>-1\text{ for all }i \in V_{-v} , \\
      0 & \text{otherwise},
    \end{cases}
  \]
  where $s = \abs{\rset{i\in V_{-v}}{p_i>0}}$ and $\pi$ ranges over all orderings of $V_{-v}$.
\end{lemma}

\begin{proof}
  We may assume that no coordinate is $0$ or $-1$, and no proper partial sum of the coordinates is $0$ or $1$, as this excludes only a finite number of lower dimensional hyperplanes.
  Fix an ordering $\pi$.
  For every $M$, membership in $B_{v,\pi,M}$ requires $0<u_i<1$ for $i\in M$ and $-1<p_{\pi_i}<0$ for $i\notin M$.
  Summing over all $M$,
  \[
    \sum_M(-1)^{\abs{M}} \ind{B_{v,\pi,M}}(p)
    =
    \prod_{i=2}^{n-1}
    \Big(
    \ind{\set{-1<p_{\pi_i}<0}}
    -
    \ind{\set{0<u_i<1}}
    \Big) .
  \]
  
  Interpret $p_{\pi_1},\ldots,p_{\pi_{n-1}}$
  as steps of a walk whose levels are
  $u_1=0,u_2,\ldots,u_{n-1},u_{n}=1$.
  As we describe next, for the product above to be nonzero, the walk must satisfy the {\em thermostat} properties: it is positive, rises below $1$ and drops above $1$.

  \begin{figure}[htbp]
    \centering
\definecolor{Blue-F}{rgb}{0,0.45,0.698}
\definecolor{Green-F}{rgb}{0,0.686,0.498}

\begin{tikzpicture}[x=1cm,y=1cm,line join=round,line cap=round,
    font={\fontsize{9pt}{11pt}\selectfont},
    every node/.style={inner sep=1pt,outer sep=0pt},
    level/.style={circle,fill=black,inner sep=1.4bp},
    rise/.style={-{Stealth[length=2mm]},draw=Green-F,text=Green-F,
      line width=1.2bp,shorten <=2bp,shorten >=2bp},
    drop/.style={-{Stealth[length=2mm]},draw=Blue-F,text=Blue-F,
      line width=1.2bp,shorten <=2bp,shorten >=2bp},
    entry/.style={inner sep=0pt,font={\fontsize{8pt}{9.5pt}\selectfont}}]
  \path[use as bounding box] (0,0) rectangle (12,4.65);
  \begin{scope}[shift={(1.7cm,.65cm)},x=1.35cm,y=2.1cm]
    \def\thSteps{1/0.3,2/0.9,3/-0.4,4/0.3,5/-0.7,6/0.6}
    \def\thStepSpacing{1}
    \def\thPlotMargin{0.4}
    \def\thListOffset{0.5}
    \def\thListGap{0.1}

    \coordinate (walk-origin) at (0,0);
    \coordinate (u1) at (walk-origin);
    \foreach \thIndex/\thValue [count=\thStep from 1,
        remember=\thTotal as \thPreviousTotal (initially 0),
        remember=\thMaximum as \thPreviousMaximum (initially 0)] in \thSteps {
      \pgfmathsetmacro{\thTotal}{round(1000*(\thPreviousTotal+\thValue))/1000}
      \pgfmathsetmacro{\thMaximum}{max(\thPreviousMaximum,\thTotal)}
      \pgfmathtruncatemacro{\thNext}{\thStep+1}
      \coordinate (u\thNext) at (\thStep*\thStepSpacing,\thTotal);
      \coordinate (walk-end) at (u\thNext);
      \coordinate (walk-top) at (0,\thMaximum);
    }
    \coordinate (band-top-left) at ($(walk-origin)+(0,1)$);
    \coordinate (walk-end-on-axis) at (walk-end |- walk-origin);
    \coordinate (plot-bottom-right) at
      ($(walk-end-on-axis)+(\thPlotMargin,0)$);
    \coordinate (plot-top-right) at (plot-bottom-right |- band-top-left);
    \coordinate (list-origin) at ($(walk-top)+(0,\thListOffset)$);

    \fill[Blue-F,fill opacity=.1] (walk-origin) rectangle (plot-top-right);
    \draw[black,line width=.8bp]
      (walk-origin) -- (plot-bottom-right) node[right,text=black] {$0$};
    \draw[black,line width=.4bp,dash pattern=on 4bp off 4bp]
      (band-top-left) -- (plot-top-right) node[right,text=black] {$1$};

    \newcommand{\thermostatList}[3]{%
      \node[entry,anchor=west,text=#3] (list-open) at (#2) {$($};
      \coordinate (list-cursor) at (list-open.east);
      \foreach \thIndex/\thValue [
          remember=\thCount as \thPreviousCount (initially 0)] in \thSteps {
        \pgfmathsetmacro{\thSignedValue}{#1*\thValue}
        \pgfmathtruncatemacro{\thCount}{\thPreviousCount+(\thSignedValue>0)}
        \ifdim\thSignedValue pt>0pt
          \node[entry,anchor=west,text=#3] (entry-\thIndex) at (list-cursor)
            {$\ifnum\thCount>1 ,\fi p_{\thIndex}=\thValue$};
          \coordinate (list-cursor) at (entry-\thIndex.east);
        \fi
      }
      \node[entry,anchor=west,text=#3] (list-close) at (list-cursor) {$)$};
      \coordinate (list-end) at (list-close.east);
    }
    \thermostatList{1}{list-origin}{Green-F}
    \coordinate (negative-list-origin) at ($(list-end)+(\thListGap,0)$);
    \thermostatList{-1}{negative-list-origin}{Blue-F}

    \node[level] at (u1) {};
    \node[below left] at (u1) {$u_1$};
    \foreach \thIndex/\thValue [count=\thStep from 1,
        remember=\thTotal as \thPreviousTotal (initially 0)] in \thSteps {
      \pgfmathsetmacro{\thTotal}{round(1000*(\thPreviousTotal+\thValue))/1000}
      \pgfmathtruncatemacro{\thNext}{\thStep+1}
      \ifdim\thValue pt>0pt
        \draw[rise] (u\thStep) --
          node[midway,above left=4bp] {$p_{\thIndex}$} (u\thNext);
      \else
        \draw[drop] (u\thStep) --
          node[midway,above right=4bp] {$p_{\thIndex}$} (u\thNext);
      \fi
      \node[level] at (u\thNext) {};
      \ifdim\thTotal pt<1pt
        \node[below=2bp] at (u\thNext) {$u_{\thNext}$};
      \else
        \node[above=2bp] at (u\thNext) {$u_{\thNext}$};
      \fi
    }
  \end{scope}
\end{tikzpicture}%
    \caption{A thermostat walk obtained by merging the ordered positive and negative coordinate lists shown above it. Take the next positive step below level $1$ and the next negative step above level $1$.}
    \label{fig:Therm}
  \end{figure}
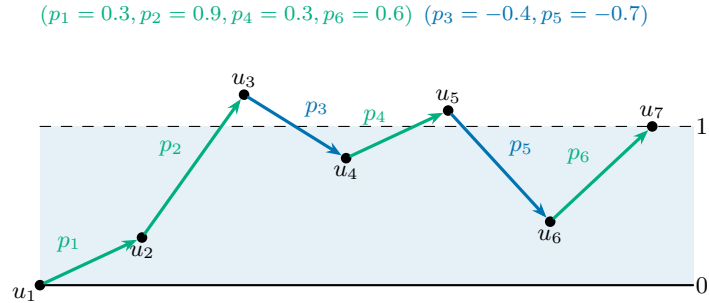
  The walk cannot reach a negative level because once it
  does, all subsequent permitted steps are negative,
  so it cannot finish at level $1$.
  It follows that inside $(0,1)$ the walk must rise (otherwise,
  it becomes negative).
  Above $1$ it must fall by less than $1$.
  The first step must be positive as well.
  In particular, all coordinates exceed $-1$.

  Now, assume that every coordinate exceeds
  $-1$. Order the $s$ positive coordinates and the
  $n-1-s$ negative coordinates separately. There is
  a unique way to merge these two lists into a walk
  of the required type; see \Cref{fig:Therm}. Take a positive step first.
  Take the next positive step whenever the
  current level is below $1$, and the next negative step whenever it is above $1$.
  The required step is always available. Indeed, the
  remaining steps sum to $1-u_i$, which is positive
  below $1$ and negative above $1$. Negative steps
  have size less than $1$, so the walk stays positive.
  The genericity assumption prevents reaching level
  $1$ before the final step.

  We obtained a bijection between nonzero products and
  pairs of orderings of the positive and negative
  coordinates. There are $s!(n-1-s)!$ pairs of orderings. 
  Every nonzero product has sign $(-1)^{s-1}$,
  one $(-1)$ factor for each positive step after
  the first.
\end{proof}

\subsection{\texorpdfstring{Proof of \Cref{lem:replacement}}{Proof of the Replacement Lemma}}
\begin{proof}
  Fix a generic $p\in H_{-v}$.
  First we prove
  \begin{equation}\label{eq:indicator-replacement}
    \sum_{\pi}\sum_M
    (-1)^{\abs{M}}\ind{B_{v,\pi,M}}
    =
    (n-1)!\sum_{\emptyset\ne R\subseteq V_{-v}}
    \frac{(-1)^{\abs{R}-1}\abs{R}}{n-\abs{R}}
    \ind{Q_{v,R}}.
  \end{equation}
  If some coordinate
  is less than $-1$, both sides of
  \eqref{eq:indicator-replacement} vanish,
  by \Cref{lem:thermostatN} and the definition
  of $Q_{v,R}$.
  Otherwise, let $s$ be the number of positive
  coordinates. The point $p$ belongs to $Q_{v,R}$
  exactly when $R$ is contained in the set of its
  positive coordinates. Thus, for each $r$, exactly $\binom sr$ sets $R$ of size $r$ contribute.
  To evaluate the right-hand side, the hockey-stick identity gives, for $1\le r\le n-1$,
  \[
    \begin{aligned}
      \sum_{j=1}^r\binom rj j!(n-1-j)!
      &=r!(n-1-r)!\sum_{j=1}^r\binom{n-1-j}{r-j}\\
      &=r!(n-1-r)!\binom{n-1}{r-1}
       =\frac{(n-1)!\,r}{n-r}.
    \end{aligned}
  \]
  By binomial inversion, the right-hand side is therefore
  \[
    (n-1)!\sum_{r=1}^s
    (-1)^{r-1}\binom sr\frac{r}{n-r}
    =(-1)^{s-1}s!(n-1-s)!.
  \]
  \Cref{lem:thermostatN} says that the left-hand side is the same.
  To conclude the proof, apply the coning part of \Cref{lem:valuation} to \cref{eq:indicator-replacement} with $q=e_v$.
\end{proof}

\begin{remark}
  Before coning, the right-hand side in \Cref{lem:replacement} is not a single simplex but a signed
  combination of dilates of the facet $\Delta_{-v}$. Its simplex content cancels, as it must because a simplex is not a signed sum of parallelepipeds. Indeed, the total coefficient of $h_{\Delta_{-v}}$ before coning is $(n-1)!\sum_{r=1}^{n-1}(-1)^{r-1}r\binom{n-1}{r} = 0$ for $n\ge3$. The identity acquires content
  only after coning with the apex $e_v$. The pyramids $\pyr{e_v}{Q_{v,R}}$
  are no longer dilates of one another.
\end{remark}

\section{Subdividing the Simplex}

This section uses the valuation property to relate the cones of the simplices $Q_{v,R}$ to the full simplex $\Delta$. Recall \Cref{fig:simPar}.
\begin{lemma}
  \label{lem:subdivision}
  For every nonempty proper subset $\emptyset \ne R \subsetneq \upto{n}$,
  \begin{equation}\label{eq:subdivision}
    \sum_{v\notin R}
    h_{\pyr{e_v}{Q_{v,R}}}
    \equiv (n-\abs{R})h_\Delta.
  \end{equation}
\end{lemma}

\begin{proof}
  Let $C \coloneqq \upto{n}\setminus R$ and
  consider the barycenter
  \[
    q_C \coloneqq \frac1{\abs{C}}\sum_{i\in C} e_i.
  \]
  For every $v\in C$, let $T_v = \pyr{q_C}{\Delta_{-v}}$.
  We claim that the simplices $T_v$ for $v \in C$ form a
  subdivision of $\Delta$ up to lower dimensional pieces; see
  \Cref{fig:direct-subdivision}.
  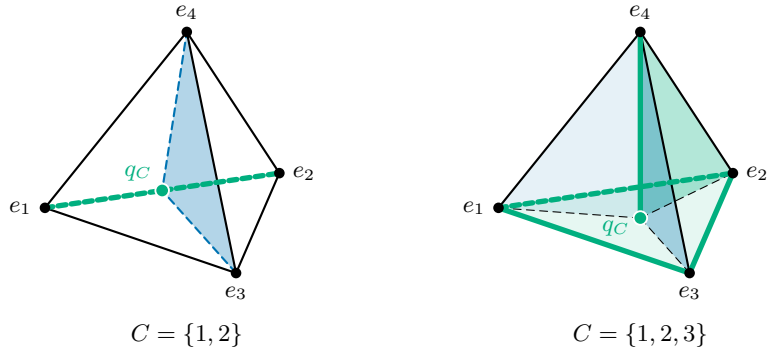
\begin{figure}[htbp]
    \centering
    \begingroup%
\definecolor{Blue-F}{rgb}{0,0.45,0.698}%
\definecolor{Green-F}{rgb}{0,0.686,0.498}%
%
\def\sdAzimuth{20}%
\def\sdElevation{24}%
\def\sdScale{1.65}%
\pgfmathsetmacro{\sdXX}{\sdScale*cos(\sdAzimuth)}%
\pgfmathsetmacro{\sdXY}{\sdScale*sin(\sdElevation)*sin(\sdAzimuth)}%
\pgfmathsetmacro{\sdYX}{-\sdScale*sin(\sdAzimuth)}%
\pgfmathsetmacro{\sdYY}{\sdScale*sin(\sdElevation)*cos(\sdAzimuth)}%
\pgfmathsetmacro{\sdZY}{\sdScale*cos(\sdElevation)}%
\tikzset{%
  sd view/.style={%
    x={(\sdXX cm,\sdXY cm)},%
    y={(\sdYX cm,\sdYY cm)},%
    z={(0cm,\sdZY cm)},%
    line join=round,line cap=round,%
    font={\fontsize{9pt}{11pt}\selectfont}},%
  sd edge/.style={draw=black,line width=.8bp},%
  sd hidden/.style={draw=black,densely dashed,line width=.4bp},%
  sd selected/.style={draw=Green-F,line width=2bp},%
  sd cut/.style={draw=Blue-F,line width=.8bp,densely dashed},%
  sd vertex/.style={circle,fill=black,inner sep=1.4bp},%
  sd barycenter/.style={circle,fill=Green-F,draw=white,
                       line width=.7bp,inner sep=1.8bp}%
}%
\newcommand{\subdivisionPoints}{%
  \coordinate (e1) at (-1,{1/sqrt(3)},0);%
  \coordinate (e2) at (1,{1/sqrt(3)},0);%
  \coordinate (e3) at (0,{-2/sqrt(3)},0);%
  \coordinate (e4) at (0,0,{sqrt(8/3)});%
  \coordinate (base-center) at (0,0,0);%
  \coordinate (edge-center) at ($(e1)!.5!(e2)$);%
  \coordinate (panel-title) at (0,0,-.9);%
}%
\newcommand{\subdivisionEdges}{%
  \draw[sd hidden] (e1) -- (e2);%
  \draw[sd edge] (e1) -- (e4) -- (e2) -- (e3) -- cycle%
                  (e3) -- (e4);%
}%
\newcommand{\subdivisionVertexLabels}{%
  \foreach \point in {e1,e2,e3,e4}%
    \node[sd vertex] at (\point) {};%
  \node[left=4bp] at (e1) {$e_1$};%
  \node[right=4bp] at (e2) {$e_2$};%
  \node[below=4bp] at (e3) {$e_3$};%
  \node[above=4bp] at (e4) {$e_4$};%
}%
%
\begin{tikzpicture}[x=1cm,y=1cm,opacity=1,
  every node/.style={font={\fontsize{9pt}{11pt}\selectfont},
                    inner sep=1pt,outer sep=0pt}]
  \path[use as bounding box] (0,0) rectangle (12,5.3);

  \begin{scope}[shift={(3.325cm,2cm)},sd view]
    \subdivisionPoints
    \coordinate (q) at (edge-center);
    \fill[Blue-F,fill opacity=.3] (q) -- (e3) -- (e4) -- cycle;
    \draw[sd cut] (e3) -- (q) -- (e4);
    \subdivisionEdges
    \draw[sd selected,densely dashed] (e1) -- (e2);
    \subdivisionVertexLabels
    \node[sd barycenter] at (q) {};
    \node[above left=4bp,text=Green-F] at (q) {$q_C$};
    \node[anchor=north,align=center] at (panel-title) {$C=\{1,2\}$};
  \end{scope}

  \begin{scope}[shift={(9.325cm,2cm)},sd view]
    \subdivisionPoints
    \coordinate (q) at (base-center);
    \fill[Green-F,fill opacity=.1] (e1) -- (e2) -- (e3) -- cycle;
    \fill[Green-F,fill opacity=.3] (q) -- (e2) -- (e4) -- cycle;
    \fill[Blue-F,fill opacity=.3] (q) -- (e3) -- (e4) -- cycle;
    \fill[Blue-F,fill opacity=.1] (q) -- (e1) -- (e4) -- cycle;
    \draw[sd hidden] (q) -- (e1) (q) -- (e2) (q) -- (e3);
    \subdivisionEdges
    \draw[sd selected,densely dashed] (e1) -- (e2);
    \draw[sd selected] (e1) -- (e3) -- (e2);
    \draw[sd selected] (q) -- (e4);
    \subdivisionVertexLabels
    \node[sd barycenter] at (q) {};
    \node[anchor=north east,text=Green-F,inner sep=0pt]
      at ([xshift=-4bp,yshift=-1bp]q) {$q_C$};
    \node[anchor=north,align=center] at (panel-title) {$C=\{1,2,3\}$};
  \end{scope}
\end{tikzpicture}%
\endgroup%
    \caption{Two subdivisions of the tetrahedron $\Delta$ into the simplices $T_v = \pyr{q_C}{\Delta_{-v}}$ for $v\in C$.}
    \label{fig:direct-subdivision}
  \end{figure}
  Indeed, consider a generic $x=\sum_{i=1}^n\lambda_i e_i\in\Delta$, where
  $\lambda_i\ge0$ and $\sum_i\lambda_i=1$. The key observation is
  \[
    x\in T_v
    \quad\iff\quad
    \lambda_v=\min_{i\in C}\lambda_i.
  \]
  Indeed, if $x=tq_C+(1-t)y$ with $y\in\Delta_{-v}$,
  then $\lambda_v=t/\abs{C}$ and $\lambda_i\ge t/\abs{C}$
  for every $i\in C$. Conversely, if $\lambda_v$ is
  minimal, take $t=\abs{C}\lambda_v$ and
  $y=(x-tq_C)/(1-t)\in\Delta_{-v}$ when $t<1$
  and if $t=1$, then $x=q_C \in T_v$.
  Thus, by \Cref{lem:valuation},
  \[
    \sum_{v\in C}h_{T_v}\equiv h_\Delta.
  \]
  By the definitions,
  \[
    \begin{aligned}
      \pyr{e_v}{Q_{v,R}}
       & =\conv\Big(
      e_v,\,
      \abs{C}\Delta_{-v}
      -\sum_{i\in C\setminus\set{v}}e_i
      \Big)                                            \\
       & =\abs{C}T_v-\sum_{i\in C\setminus\set{v}}e_i ,
    \end{aligned}
  \]
  because $e_v = \abs{C} q_C - \sum_{i\in C\setminus\set{v}} e_i$.
  Translation changes a support function by a
  linear function, which is $\equiv 0$.
  Consequently,
  \[
    h_{\pyr{e_v}{Q_{v,R}}}
    \equiv \abs{C}h_{T_v}.
  \]
  Summing over $v\in C$ and applying the above gives
  \[
    \sum_{v\notin R} h_{\pyr{e_v}{Q_{v,R}}}
    \equiv
    \abs{C}h_\Delta
    =
    (n-\abs{R}) h_\Delta.
  \]
\end{proof}

\section{Putting it Together}
The following theorem expresses the simplex as a signed sum of paramids modulo lower dimensional pieces.

\begin{theorem}
  \[h_\Delta \equiv \frac{(-1)^n}{ n!}
    \sum_{v=1}^n
    \sum_{\pi}
    \sum_{M\subseteq\set{2,\ldots,n-1}}
    (-1)^{\abs{M}}
    h_{\pyr{e_v}{B_{v,\pi,M}}}.\]
\end{theorem}

\begin{proof}
  Summing the identity from \Cref{lem:replacement} over all
  $v\in\upto{n}$ gives
  \[
    \begin{aligned}
      \sum_v\sum_{\pi}\sum_M
      (-1)^{\abs{M}}
      h_{\pyr{e_v}{B_{v,\pi,M}}}
      \equiv
      (n-1)!\sum_{\emptyset\ne R\subsetneq \upto{n}}
      \frac{(-1)^{\abs{R}-1}\abs{R}}{n-\abs{R}}
      \sum_{v\notin R}
      h_{\pyr{e_v}{Q_{v,R}}}.
    \end{aligned}
  \]
  By \Cref{lem:subdivision}, this equals (with $\equiv$)
  \[
    (n-1)!\sum_{\emptyset\ne R \subsetneq \upto{n}}
    (-1)^{\abs{R}-1}\abs{R}\,h_\Delta=(n-1)!\left(
    \sum_{r=1}^{n-1}
      (-1)^{r-1}r\binom nr
    \right)h_\Delta =
    (-1)^n n!\,h_\Delta.
  \]
\end{proof}

\printbibliography

@article{shephard1974combinatorial,
  author  = {Shephard, G. C.},
  title   = {Combinatorial Properties of Associated Zonotopes},
  journal = {Canadian Journal of Mathematics},
  volume  = {26},
  number  = {2},
  pages   = {302--321},
  year    = {1974},
  doi     = {10.4153/CJM-1974-032-5}
}

@article{KPS2026,
  author = {Koutschan, Christoph and Ponomarchuk, Anton and Schicho, Josef},
  journal = {IEEE Transactions on Information Theory},
  title = {Representing Piecewise-Linear Functions by Functions With Minimal Arity},
  year = {2026},
  volume = {72},
  number = {9},
  pages = {6825-6831},
  doi = {10.1109/TIT.2026.3705562}
}

@article{ruess2026shallower,
  title = {Shallower {ReLU} network representations via exact linear algebra},
  author = {Rue{\ss}, Kilian and Averkov, Gennadiy and Brunck, Florestan and Grillo, Moritz and Hertrich, Christoph and Loho, Georg and Stade, Jack and Stargalla, Moritz and Sun, Matthew and Winter, Martin},
  journal = {arXiv preprint arXiv:2607.21651},
  year = {2026},
  eprint = {2607.21651},
  archiveprefix = {arXiv},
  doi = {10.48550/arXiv.2607.21651},
  url = {https://arxiv.org/abs/2607.21651v2},
  note = {Version 2, 1 September 2026}
}

@article{KMPS2025,
  author = {Koutschan, Christoph and Moser, Bernhard and Ponomarchuk, Anton and Schicho, Josef},
  date = {2025-07-01},
  doi = {10.1007/s00200-023-00627-1},
  id = {Koutschan2025},
  journal = {Applicable Algebra in Engineering, Communication and Computing},
  number = {4},
  pages = {595--610},
  title = {Representing piecewise linear functions by functions with small arity},
  url = {https://doi.org/10.1007/s00200-023-00627-1},
  volume = {36},
  year = {2025},
  issn = {1432-0622},
  ids = {koutschanRepresentingPiecewiseLinear2025}
}

@article{JT78,
  issn = {00255521, 19031807},
  url = {http://www.jstor.org/stable/24491087},
  author = {B{\o}rge Jessen and Anders Thorup},
  journal = {Mathematica Scandinavica},
  number = {2},
  pages = {211--240},
  publisher = {Mathematica Scandinavica},
  title = {THE ALGEBRA OF POLYTOPES IN AFFINE SPACES},
  urldate = {2026-10-11},
  volume = {43},
  year = {1978}
}

@book{Weibel1994,
  place = {Cambridge},
  series = {Cambridge Studies in Advanced Mathematics},
  title = {An Introduction to Homological Algebra},
  publisher = {Cambridge University Press},
  author = {Weibel, Charles A.},
  year = {1994},
  collection = {Cambridge Studies in Advanced Mathematics}
}

@incollection{MS1983,
  title = {Valuations on convex bodies},
  author = {McMullen, Peter and Schneider, Rolf},
  booktitle = {Convexity and its applications},
  pages = {170--247},
  year = {1983},
  publisher = {Springer}
}

@inbook{Schneider2013,
  place = {Cambridge},
  series = {Encyclopedia of Mathematics and its Applications},
  title = {Valuations on convex bodies},
  booktitle = {Convex Bodies: The Brunn--Minkowski Theory},
  publisher = {Cambridge University Press},
  author = {Schneider, Rolf},
  year = {2013},
  pages = {329--368},
  collection = {Encyclopedia of Mathematics and its Applications}
}

@inproceedings{aroraUnderstandingDeepNeural2018,
  title = {Understanding Deep Neural Networks with Rectified Linear Units},
  booktitle = {International {{Conference}} on {{Learning Representations}}},
  author = {Arora, Raman and Basu, Amitabh and Mianjy, Poorya and Mukherjee, Anirbit},
  year = {2018},
  url = {https://openreview.net/forum?id=B1J_rgWRW},
  eprint = {1611.01491},
  archiveprefix = {arXiv},
  ids = {arora2018understanding}
}

@inproceedings{bakaevBetterNeuralNetwork2026,
  title = {Better Neural Network Expressivity: Subdividing the Simplex},
  booktitle = {Proceedings of the 58th {{Annual ACM Symposium}} on {{Theory}} of {{Computing}}},
  author = {Bakaev, Egor and Brunck, Florestan and Hertrich, Christoph and Stade, Jack and Yehudayoff, Amir},
  year = {2026},
  publisher = {ACM},
  address = {Salt Lake City, UT, USA},
  doi = {10.1145/3798129.3800768},
  pages = {500--507},
  url = {https://arxiv.org/abs/2505.14338},
  eprint = {2505.14338},
  archiveprefix = {arXiv},
  note = {Preprint first posted in 2025},
  ids = {bakaev2026better}
}

@inproceedings{haaseLowerBoundsDepth2023,
  title = {Lower Bounds on the Depth of Integral {{ReLU}} Neural Networks via Lattice Polytopes},
  booktitle = {International {{Conference}} on {{Learning Representations}}},
  author = {Haase, Christian Alexander and Hertrich, Christoph and Loho, Georg},
  year = {2023},
  url = {https://openreview.net/forum?id=2mvALOAWaxY},
  eprint = {2302.12553},
  archiveprefix = {arXiv},
  ids = {haase2023lower}
}

@article{hertrichLowerBoundsDepth2023,
  title = {Towards Lower Bounds on the Depth of {{ReLU}} Neural Networks},
  author = {Hertrich, Christoph and Basu, Amitabh and Di Summa, Marco and Skutella, Martin},
  year = {2023},
  journal = {SIAM Journal on Discrete Mathematics},
  volume = {37},
  number = {2},
  doi = {10.1137/22M1489332}
}

@article{leshnoMultilayerFeedforwardNetworks1993,
  title = {Multilayer Feedforward Networks with a Nonpolynomial Activation Function Can Approximate Any Function},
  author = {Leshno, Moshe and Lin, Vladimir Ya. and Pinkus, Allan and Schocken, Shimon},
  year = {1993},
  journal = {Neural Networks},
  volume = {6},
  number = {6},
  doi = {10.1016/S0893-6080(05)80131-5},
  pages = {861--867},
  ids = {leshno1993multilayer}
}

@article{he2020relu,
  title={ReLU deep neural networks and linear finite elements},
  author={He, Juncai and Li, Lin and Xu, Jinchao and Zheng, Chunyue},
  journal={Journal of Computational Mathematics},
  volume={38},
  number={3},
  pages={502--527},
  year={2020},
  publisher={JSTOR}
}

@inproceedings{amos2017input,
  author = {Amos, Brandon and Xu, Lei and Kolter, J. Zico},
  title = {Input Convex Neural Networks},
  booktitle = {Proceedings of the 34th International Conference on Machine Learning},
  series = {Proceedings of Machine Learning Research},
  volume = {70},
  pages = {146--155},
  year = {2017},
  publisher = {PMLR},
  url = {https://proceedings.mlr.press/v70/amos17b.html}
}

@inproceedings{bunning2021input,
  author = {B{\"u}nning, Felix and Schalbetter, Adrian and Aboudonia, Ahmed and {Hudoba de Badyn}, Mathias and Heer, Philipp and Lygeros, John},
  title = {Input Convex Neural Networks for Building {MPC}},
  booktitle = {Proceedings of the 3rd Conference on Learning for Dynamics and Control},
  series = {Proceedings of Machine Learning Research},
  volume = {144},
  pages = {251--262},
  year = {2021},
  publisher = {PMLR},
  url = {https://proceedings.mlr.press/v144/bunning21a.html}
}

@inproceedings{chen2018optimal,
  author = {Chen, Yize and Shi, Yuanyuan and Zhang, Baosen},
  title = {Optimal Control Via Neural Networks: A Convex Approach},
  booktitle = {International Conference on Learning Representations},
  year = {2019},
  url = {https://openreview.net/forum?id=H1MW72AcK7},
  eprint = {1805.11835},
  archiveprefix = {arXiv},
  note = {Preprint first posted in 2018; conference publication in 2019}
}

@article{chen2020data,
  author = {Chen, Yize and Shi, Yuanyuan and Zhang, Baosen},
  title = {Data-Driven Optimal Voltage Regulation Using Input Convex Neural Networks},
  journal = {Electric Power Systems Research},
  volume = {189},
  pages = {106741},
  year = {2020},
  doi = {10.1016/j.epsr.2020.106741},
  url = {https://www.sciencedirect.com/science/article/pii/S0378779620305447}
}

@article{bunne2023learning,
  author = {Bunne, Charlotte and Stark, Stefan G. and Gut, Gabriele and {Sarabia del Castillo}, Jacobo and Levesque, Mitch and Lehmann, Kjong-Van and Pelkmans, Lucas and Krause, Andreas and R{\"a}tsch, Gunnar},
  title = {Learning Single-Cell Perturbation Responses Using Neural Optimal Transport},
  journal = {Nature Methods},
  volume = {20},
  number = {11},
  pages = {1759--1768},
  year = {2023},
  doi = {10.1038/s41592-023-01969-x},
  url = {https://www.nature.com/articles/s41592-023-01969-x}
}

@misc{hertrich2024neural,
  author = {Hertrich, Christoph and Loho, Georg},
  title = {Neural Networks and (Virtual) Extended Formulations},
  year = {2024},
  eprint = {2411.03006},
  archiveprefix = {arXiv},
  doi = {10.48550/arXiv.2411.03006},
  url = {https://arxiv.org/abs/2411.03006},
  note = {First posted in 2024; revised in 2025}
}

@article{bengio2009learning,
  author = {Bengio, Yoshua},
  title = {Learning Deep Architectures for {AI}},
  journal = {Foundations and Trends in Machine Learning},
  volume = {2},
  number = {1},
  pages = {1--127},
  year = {2009},
  doi = {10.1561/2200000006}
}

@article{mhaskar2017when,
  author = {Mhaskar, Hrushikesh and Liao, Qianli and Poggio, Tomaso},
  title = {When and Why Are Deep Networks Better Than Shallow Ones?},
  volume = {31},
  number = {1},
  year = {2017},
  doi = {10.1609/aaai.v31i1.10913},
  url = {https://ojs.aaai.org/index.php/AAAI/article/view/10913},
  journal = {Proceedings of the AAAI Conference on Artificial Intelligence}
}

@article{huchette2026when,
  author = {Huchette, Joey and Mu{\~n}oz, Gonzalo and Serra, Thiago and Tsay, Calvin},
  title = {When Deep Learning Meets Polyhedral Theory: A Survey},
  journal = {INFORMS Journal on Computing},
  year = {2026},
  doi = {10.1287/ijoc.2024.0902},
  url = {https://pubsonline.informs.org/doi/10.1287/ijoc.2024.0902},
  eprint = {2305.00241},
  archiveprefix = {arXiv},
  note = {Published online 12 March 2026; preprint first posted in 2023}
}

@inproceedings{eldan2016power,
  author = {Eldan, Ronen and Shamir, Ohad},
  title = {The Power of Depth for Feedforward Neural Networks},
  booktitle = {29th Annual Conference on Learning Theory},
  series = {Proceedings of Machine Learning Research},
  volume = {49},
  pages = {907--940},
  year = {2016},
  publisher = {PMLR},
  url = {https://proceedings.mlr.press/v49/eldan16.html}
}

@inproceedings{telgarsky2016benefits,
  author = {Telgarsky, Matus},
  title = {Benefits of Depth in Neural Networks},
  booktitle = {29th Annual Conference on Learning Theory},
  series = {Proceedings of Machine Learning Research},
  volume = {49},
  pages = {1517--1539},
  year = {2016},
  publisher = {PMLR},
  url = {https://proceedings.mlr.press/v49/telgarsky16.html}
}

@inproceedings{safran2022optimization,
  author = {Safran, Itay and Lee, Jason D.},
  title = {Optimization-Based Separations for Neural Networks},
  booktitle = {Proceedings of Thirty Fifth Conference on Learning Theory},
  series = {Proceedings of Machine Learning Research},
  volume = {178},
  pages = {3--64},
  year = {2022},
  publisher = {PMLR},
  url = {https://proceedings.mlr.press/v178/safran22a.html}
}

@inproceedings{dai2021lyapunov,
  author = {Dai, Hongkai and Landry, Benoit and Yang, Lujie and Pavone, Marco and Tedrake, Russ},
  title = {Lyapunov-Stable Neural-Network Control},
  booktitle = {Proceedings of Robotics: Science and Systems},
  year = {2021},
  doi = {10.15607/RSS.2021.XVII.063},
  url = {https://roboticsproceedings.org/rss17/p063.html}
}

@article{cybenko1989approximation,
  author = {Cybenko, George},
  title = {Approximation by Superpositions of a Sigmoidal Function},
  journal = {Mathematics of Control, Signals and Systems},
  volume = {2},
  number = {4},
  pages = {303--314},
  year = {1989},
  doi = {10.1007/BF02551274},
  url = {https://link.springer.com/article/10.1007/BF02551274}
}

@article{hornik1989multilayer,
  author = {Hornik, Kurt and Stinchcombe, Maxwell and White, Halbert},
  title = {Multilayer Feedforward Networks Are Universal Approximators},
  journal = {Neural Networks},
  volume = {2},
  number = {5},
  pages = {359--366},
  year = {1989},
  doi = {10.1016/0893-6080(89)90020-8},
  url = {https://www.sciencedirect.com/science/article/pii/0893608089900208}
}

@inproceedings{glorot2011deep,
  author = {Glorot, Xavier and Bordes, Antoine and Bengio, Yoshua},
  title = {Deep Sparse Rectifier Neural Networks},
  booktitle = {Proceedings of the Fourteenth International Conference on Artificial Intelligence and Statistics},
  series = {Proceedings of Machine Learning Research},
  volume = {15},
  pages = {315--323},
  year = {2011},
  publisher = {PMLR},
  url = {https://proceedings.mlr.press/v15/glorot11a.html}
}
\end{document}